\documentclass[aps,onecolumn]{IEEEtran}

\usepackage{amsmath,amsthm,amssymb}
\usepackage[a4paper, total={6.5in, 9.5in}]{geometry}

\usepackage{enumerate}
\usepackage{color}
\usepackage{xcolor}
\usepackage{url}
\usepackage{balance}
\usepackage{graphicx} 
\usepackage{mathrsfs}
\usepackage{epsfig}
\usepackage{verbatim}
\usepackage{setspace}
\usepackage{cite}
\usepackage{multicol}
\usepackage{lipsum}
\usepackage{etoolbox}
\usepackage{algpseudocode}
\usepackage{algorithmicx}
\usepackage{algorithm}
\usepackage{hyperref}

\newtheorem{theorem}{Theorem}%[section]
\newtheorem{lemma}{Lemma}

\newtheorem{corollary}{Corollary}
\newtheorem{definition}{Definition}
\newtheorem{claim}{Claim}
\newtheorem{remark}{Remark}

\newcommand{\bX}{\mathbf{X}}

\newcommand{\bx}{\mathbf{x}}
\newcommand{\bY}{\mathbf{Y}}

\newcommand{\by}{\mathbf{y}}
\newcommand{\bz}{\mathbf{z}}

\newcommand{\bbP}{\mathbb{P}}
\newcommand{\bbR}{\mathbb{R}}

\newcommand{\cX}{{\cal X}}
\newcommand{\cY}{{\cal Y}}
\newcommand{\cZ}{{\cal Z}}
\newcommand{\cS}{{\cal S}}
\newcommand{\cC}{{\cal C}}
\newcommand{\cL}{{\cal L}}
\newcommand{\cB}{{\cal B}}

\newcommand{\cP}{{\cal P}}

\newcommand{\cA}{{\cal A}}

\newcommand{\view}{V}

\newcommand{\td}[1]{\tilde{#1}}

\newcommand{\rx}{X}

\newcommand{\removed}[1]{}

\newcommand{\tdc}{{\td{c}}}
\newcommand{\tdx}{{\td{x}}}
\newcommand{\tdbx}{{\td{\bx}}}

\newcommand{\e}{\epsilon}

\newcommand{\gunc}[1]{Gaussian UNC[$\gamma^2$,$\delta^2$] }
\newcommand{\pgunc}[1]{Passive-Gaussian UNC[$\gamma^2$,$\delta^2$] }

\makeatletter
\patchcmd{\@IEEEeqnarray}{\relax}{\relax\intertext@}{}{}
\makeatother

\algdef{SE}[DOWHILE]{Do}{doWhile}{\algorithmicdo}[1]{\algorithmicwhile\ #1}%

\newcommand{\gsq}{\gamma^2}
\newcommand{\dsq}{\delta^2}
\newcommand{\tsq}{\theta^2}

\newcommand{\cO}{\mathcal{O}}
\newcommand{\cN}{\mathcal{N}}
\newcommand{\cG}{\mathcal{G}}

\newcommand{\tdcA}{\tilde{\cA}}
\newcommand{\tdcB}{\tilde{\cB}}
\newcommand{\tdX}{\tilde{X}}

\newcommand{\scrP}{\mathscr{P}}

\newcommand{\guncw}{{{Gaussian UNC}}}
\newcommand{\GUNC}{{{\fontfamily{lmtt}\selectfont{GaussUNC}}}}
\newcommand{\SGUNC}{{{\fontfamily{lmtt}\selectfont{SimGaussUNC}}}}
\newcommand{\gsqdsq}{$[\gsq,\dsq]$}

\title{Commitment over Gaussian Unfair Noisy Channels}
\author{
\IEEEauthorblockN{Amitalok J. Budkuley\textbf{*}, Pranav Joshi\textbf{*}, Manideep Mamindlapally\textbf{*} and Anuj Kumar Yadav\textbf{*}{{$^{\href{https://orcid.org/0000-0001-7763-3787}{\includegraphics[scale=0.08]{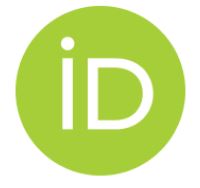}}}$}}}\\
\thanks{
\noindent \textbf{Amitalok J. Budkuley} is with the Department of Electronics and Electrical Communication Engineering, Indian Institute of Technology Kharagpur, West Bengal, India (Email: \texttt{amitalok@ece.iitkgp.ac.in}). 

\noindent \textbf{Pranav Joshi} is an independent researcher (Email: \texttt{pranavjoshi@iitkgp.ac.in}).

\noindent \textbf{Manideep Mamindally} is with the Tata Institute of Fundamental Research (TIFR), Mumbai, India (Email: \texttt{manideepyx@iitkgp.ac.in}).

\noindent \textbf{Anuj Kumar Yadav} is with the Laboratory for Information in Networked Systems( LINX), School of Computer and Communication Sciences, École Polytechnique Fédérale de Lausanne (EPFL), Switzerland (Email: \texttt{anuj.yadav@epfl.ch}).\\
}}
\removed{\author
\IEEEauthorblockN{Amitalok Budkuley*{{$^{\href{https://orcid.org/0000-0001-7763-3787}{\includegraphics[scale=0.08]{orcid.JPG}}}$}}}
        \\\IEEEauthorblockA{IIT Kharagpur \\ amitalok.ece@iitkgp.ac.in}
		\and
  {  \IEEEauthorblockN{Pranav Joshi*{{$^{\href{https://orcid.org/0000-0001-7763-3787}{\includegraphics[scale=0.08]{orcid.JPG}}}$}}}\\
    \IEEEauthorblockA{Independent Researcher\\ pranavjoshi@iitkgp.ac.in}
	}	\and
    \IEEEauthorblockN{Manideep Mamindlapally*{{$^{\href{https://orcid.org/0000-0001-7763-3787}{\includegraphics[scale=0.08]{orcid.JPG}}}$}}}\\
    \IEEEauthorblockA{TIFR \\ manideepyx@iitkgp.ac.in}
		  \and
	  \IEEEauthorblockN{Anuj Kumar Yadav*{{$^{\href{https://orcid.org/0000-0001-7763-3787}{\includegraphics[scale=0.08]{orcid.JPG}}}$}}}\\
    \IEEEauthorblockA{EPFL \\ anuj.yadav@epfl.ch}    
}
\begin{document}
\maketitle 
\maketitle 
{\let\thefootnote\relax\footnotetext{\textbf{{($*$) : The paper follows alphabetical author order. {PJ}, MM, and AKY have equally contributed to this work.}}}}
\begin{abstract}
%
%\ajb{reduce abstract}
Commitment is a key primitive which resides at the heart of several cryptographic protocols. Noisy channels can help realize information-theoretically secure commitment schemes; however, their imprecise statistical characterization can severely impair such schemes, especially their security guarantees. Keeping our focus on channel `unreliability' in this work, we study commitment over unreliable continuous alphabet channels called the Gaussian unfair noisy channels or Gaussian UNCs. \\

We present the first results on the optimal throughput or commitment capacity of Gaussian UNCs. It is known that `classical' Gaussian channels have infinite commitment capacity, even under finite transmit power constraints. For `unreliable' Gaussian UNCs, we prove the surprising result that their commitment capacity may be finite, and in some cases, zero. When commitment is possible, we present achievable rate lower bounds by constructing positive-throughput protocols under given input power constraint, and (two-sided) channel elasticity at committer Alice and receiver Bob. 
%; however, there is a gap in our positive rate threshold characterization in general. O
%Our protocols reveal a direct relation between commitment rate and two important parameters, viz., Alice’s power constraint and the (two-sided) UNC elasticity at Alice and Bob. 
%We present non-trivial achievable rate bounds for commitment capacity under given input power constraints, and channel elasticity. 
Our achievability results establish an interesting fact – Gaussian UNCs with zero elasticity have infinite commitment capacity - which brings a completely new perspective to why classic Gaussian channels, i.e., Gaussian UNCs with zero elasticity, have infinite capacity. Finally, we precisely characterize the positive commitment capacity threshold for a Gaussian UNC in terms of the channel elasticity, when the transmit power tends to  infinity.
\end{abstract}
%
%\begin{IEEEkeywords}
%Commitment capacity, reverse elastic channels, unreliable channels, randomness extractors, information-theoretic security.
%\end{IEEEkeywords}
%
\newpage
\tableofcontents
\section{Introduction}\label{sec:introduction}
%\textcolor{red}{Extended draft:} \url{https://bit.ly/2RkOeoD} \\

%\footnote{The author order is alphabetic. This work was partially supported by ISIRD, IIT Kharagpur.}
%
Commitment is a widely studied cryptographic primitive. In essence, commitment can be seen as a successful realization of a sealed envelope exchange between two mutually distrustful parties. A \emph{committer} Alice seeks to \emph{commit} to a message $c.$ She puts this message $c$ in a sealed envelope, and distributes it publicly to \emph{receiver} Bob. At a later time of her choosing, Alice opens the envelope to \emph{reveal} her message to Bob. In the sealed envelope exchange, both parties seek  the following guarantees: Alice seeks a \emph{hiding} guarantee where Bob is unable to learn the contents of the sealed envelope  prior to her act  of opening it. Bob seeks a \emph{binding} guarantee where the sealed contents can be verified by him to be tamper-proof, i.e., Alice cannot successfully cheat by revealing a different message. Given the nature of its functionality, commitment is a useful building block in several non-trivial cryptographic protocols, including secure multiparty computation~\cite{cramer2015secure}.  

Wyner's work on wiretap channels~\cite{wyner1} brought to the fore the value of noisy channels in realizing information-theoretic security.\footnote{In his pioneering work, Blum~\cite{blum} studied conditionally-secure commitment where parties are computationally bounded.} Cr{\'e}peau~\cite{crepeau_achieving_1988-1} first demonstrated the possibility of information-theoretic secure commitment over binary noisy channels. Commitment has subsequently been widely studied over several noisy channels. Winter \emph{et al.}~\cite{winter2003commitment} completely characterized the maximum commitment throughput or \emph{commitment capacity} over any discrete memoryless channels (DMCs). This result was extended to DMCs with arbitrary cost constraints in~\cite{mymb}, where an alternate dual commitment capacity characterization was also presented. Nascimento \emph{et al.} studied commitment over continuous alphabet additive white Gaussian noise (AWGN) channels in~\cite{nascimento-barros-t-it2008}. They showed the surprising result that the commitment capacity over any non-trivial  AWGN channel is infinite, even under a finite power constraint at the committer. Subsequently,~\cite{oggier2008practical} showed a constructive commitment scheme (using lattice codes) over AWGN channels.

Despite a noisy channel being an excellent resource for realizing commitment, one needs to be mindful of potential \emph{unreliability} in a channel-- oftentimes due to imprecise channel characterization (passive unreliability) or due to channel tampering by malicious parties (active unreliability). Channel unreliabilities impact the commitment throughput potential of a noisy channel;\footnote{In this work, we follow the \emph{game-based} security paradigm which is  different from an alternate   \emph{simulation-based} paradigm (cf.~\cite{canetti2001universally}). In the absence of computational limitations, however, it is known that the game-based security notion coincides with the simulation-based security notion.} in some cases, they may completely preclude commitment. Damg{\aa}rd \emph{et al.}~\cite{damgaard1999possibility} first studied channel unreliability via the \emph{unfair noisy channels} (UNC) over the binary alphabet. Their \emph{Binary UNC} combines both forms of unreliability, \textit{viz.,} active unreliability (when either of two parties are malicious) and passive unreliability (when both parties are honest). In a classic result,  Damg{\aa}rd \emph{et al.}~\cite{damgaard1999possibility} precisely characterized the threshold for positive-rate commitment over Binary UNCs; their commitment capacity was recently characterized in~\cite{crepeau2020commitment}, albeit under some assumptions.\footnote{The authors in~\cite{crepeau2020commitment} impose  restrictions on protocols for their converse.} Other works, for instance~\cite{wullschleger2009oblivious,khurana2016secure,budkuley2021commitment,budkuley2022reverse,yadav2021commitment,yadav2022commitment,budkuley2022commitment}, have studied other related models like the elastic/reverse elastic channels and compound channels respectively where the noisy channels exhibit either of these unreliabilities exclusively. %Other forms of channel unreliabilities have also been studied~\cite{wullschleger2009oblivious}.

The focus of our study is \emph{unreliable} Gaussian channels; to the best of our knowledge, this work presents the first results on their commitment capacity.\footnote{Although the authors in~\cite{nascimento-barros-t-it2008} principally study commitment over the classic AWGN channel, they have a short discussion on such unfair noisy channels, where they argue (without any formal proof) that their commitment throughput may be finite.} In particular, we study the \emph{Gaussian unfair noisy channel} or the \emph{Gaussian-UNC} with parameters $\gamma^2,\delta^2,$ where $0<\gamma^2\leq \delta^2;$ the definition follows. 

\begin{definition}[Gaussian Unfair Noisy Channel (Gaussian-UNC)]\label{def:gaussian:UNC}
A Gaussian Unfair Noisy Channel (GaussianUNC) with parameters  $0<\gamma^2\leq\delta^2$, also called Gaussian-UNC$[\gamma^2,\delta^2]$, is a noisy AWGN channel where (i) honest parties communicate over an AWGN channel where the noise variance can take values in the set $\cS=[\gamma^2,\delta^2]$ and is unknown to them, (ii) any dishonest party can privately set the noise variance to any value in $\cS$.\footnote{When $\gamma^2=\delta^2,$ the Gaussian UNC$[\gamma^2,\delta^2]$ specializes to an AWGN channel with i.i.d. zero mean Gaussian noise with a fixed variance $\gamma^2=\delta^2.$ The commitment capacity of such a channel is known to be infinite~\cite{nascimento-barros-t-it2008}. Furthermore, when $\gamma^2=\delta^2=0,$ the AWGN channel specializes to a noiseless channel whose commitment capacity is zero~\cite{blum}. Hence, $\gamma$, $\delta$ are so chosen.} Here, $E:=\delta^2-\gamma^2$ is called the  elasticity of the channel.
\end{definition} 
In such a Gaussian UNC$[\gamma^2,\delta^2],$ the independent and identically distributed (i.i.d.) noise random variable is a zero-mean white Gaussian noise whose variance can vary in the range $[\gamma^2,\delta^2].$ Furthermore, as in Binary UNCs, only malicious parties can privately set the variance to any value in the given range; honest parties are assumed to be `blind' to the instantiated variance value.\\ 

\indent\textbf{Contributions:}
The following are key contributions:
\begin{itemize}
\item We show that commitment is impossible over a Gaussian UNC[$\gamma^2,\delta^2$] under unconstrained inputs when $\delta^2\geq 2\gamma^2$ (cf.~Theorem \ref{thm:impos}). More generally, our impossibility characterization over Gaussian UNCs extends to some other secure multiparty functionalities, like oblivious transfer (cf. Corollary~\ref{cor:smc}).
%Our zero-rate converse is inspired by that for Binary UNCs~~\cite{damgaard1999possibility} but the specific analysis differs owing to continuous alphabet of the channel. \\
\item We also present partial lower bounds on the commitment capacity (cf. Theorem~\ref{thm:possible:highSER}). Our results show that channels with non-zero elasticity at both Alice and Bob, may have \emph{finite} rate. 
%We motivate a novel notion of \emph{signal-to-elasticty ratio (SER)} for Gaussian UNCs (under Alice's power constraint) and show that achievability rate is intimately connected to whether $SER>1$ or $SER\leq 1.$  (cf.~Theorems~\ref{thm:possible:highSER} and~\ref{thm:possible:lowSER}). 
A key takeaway from this work is the  precise characterization of the positive rate threshold when $P\rightarrow \infty$  for fixed $\gamma^2,\delta^2$ (cf. Theorem~\ref{thm:achieve:inftyP}). 
\item Our results present a new perspective, via elasticity, on the infinite capacity of classic AWGN channels. In particular, we show that the zero elasticity (at Alice and Bob) of the classic AWGN is the key reason for their infinite capacity.     
%We then present a conjecture on the commitment capacity of its \emph{symmetric} channel instance, viz., the symmetric two-sided elastic channel. 
\end{itemize}
\vspace{2mm}

\noindent\textbf{Organization of paper:} In the following section~\ref{sec:not:pre}, we briefly present the basic notation following which we present our  problem setup in Section~\ref{sec:system:model}. 
%This is followed by Section~\ref{sec:system:model}, where we describe our problem setup. 
The main results of this work  are presented in Section~\ref{sec:main:results}. We present the details of the converse proof and overview of achievability in Section~\ref{sec:proofs}.
Finally, we present the achievability protocol and prove its security in the Appendix.

\section{Notation and Preliminaries}\label{sec:not:pre}
%\subsection{Notation}
%
Random variables are denoted by upper case letters (eg. $X$), their values  by lower case letters (eg., $x$), and their alphabets by calligraphic letters (eg. $\cX$). %Unless stated otherwise, all sets are finite. 
Random vectors and their accompanying values are denoted by boldface letters. 
%(e.g., $\bX=(X_1,X_2,\cdots,X_n)$, $\bx=(x_1,x_2,\cdots,x_n)$, resp.).
%Here $n$ denotes the block length of communication.
%The set of real numbers, non-negative real numbers and real vectors (of length $n$) are denoted by $\mathbb{R}$, $\mathbb{R}_+$, and $\mathbb{R}^n$ respectively. 
%The set of natural numbers is denoted by $\mathbb{N}$. 
For  natural number $a\in\mathbb{N}$, let $[a]:=\{1,2,\cdots, a\}$. %Let $\bX^i=(X_1,X_2,\cdots,X_i)$ and $\bX_i^j=(x_i,X_{i+1},\cdots,X_{j})$ denote vectors. 
%We denote the Hamming distance between two vectors, say $\vec{x},\vec{x}'\in\cX^n$ by $d_H(\vec{x},\vec{x}').$
%$=\sum_{i=1}^n \mathbf{1}_{\{x_i\neq x'_i\}}$
%, where $\mathbf{1}_{A}$  denotes the indicator of $A$.
%
%Let $P_X$ denote the distribution of $X\in\cX$; $\cP(\cX)$ denotes the simplex of distributions on $\cX$. Distributions for multiple random variables are similarly defined. 
%Let $\cP(\cX|\cY)$ denote the set of all conditional probability distributions on random variable $X\in\cX$ conditioned on $Y\in\cY$. We denote by $P_X$, $P_{X|Y}$ and $P_{X,Y}$ the probability distribution on random variables $X\in\cX$, conditional probability distribution on random variable $X\in\cX$ conditioned on random variable $Y\in\cY$ and joint probability distribution  on the pair of random variables $(x,Y)\in\cX\times \cY$. For the latter, we denote the marginal distribution on random variable $X$ by $[P_{X,Y}]_X$. 
%Given $P_{\rx}$, $P_{\rx}^{(n)}$ denotes the $n$-fold memoryless extension of $P_\rx$. 
%
Let $\mathbb{P}(A)$ denote the probability of event $A$. Deterministic (resp. random) functions will be denoted by lower case (resp. upper case) letters. Let $X\sim \text{Unif}(\cX)$ and $Y\sim \text{Bernoulli}(p)$ denote a uniform random variable over $\cX$ and a Bernoulli random variable $X$ with parameter $p\in[0,1]$. %Let $p*q:=p(1-q)+(1-p)q$, where $p,q\in[0,1].$ 
Given $P_\rx,Q_\rx \in \cP(\cX)$, 
SD$(P_{\rx},Q_\rx)$ denotes their statistical distance.
%
%Next, we define some classic information measures. 

Let random variables $X,Y\in\cX\times\cY$, where $(X,Y)\sim P_{X,Y}$. The \emph{min-entropy} of $X$ is denoted by $H_{\infty}(X):=\min_{x\in\cX} \left(-\log(P_{X}(x))\right)$; the conditional version is given by $H_{\infty}(X|Y):=\min_{y} H_{\infty}(X|Y=y).$ For $\epsilon\in[0,1)$, the \emph{$\epsilon$-smooth min entropy} and its conditional version is given by:
$H_{\infty}^\epsilon(X):= \max_{X':||P_{X'}-P_{X}||\leq\hspace{1mm}\epsilon} H_{\infty}(X')$ and 
$H_{\infty}^\epsilon(X|Y):= \max_{X',Y':||P_{X',Y'}-P_{X,Y}||\leq\hspace{1mm}\epsilon} H_{\infty}(X'|Y')$ respectively. While the above notions are specified for discrete variables, similar results hold to continuous alphabets ; the discretization approach is used to extend some of these definitions (see appendix for details). We will present some of these later whenever necessary.

\removed{
Let $X \in \cX$ and $Y \in \mathbb{R}$ represent a discrete random variable and a continuous random variable respectively. For every $x \in \cX$, the conditional probability density function (PDF) $f_{Y|X}(y|x)$ is assumed to be Riemann integrable. Then, the PDF of $Y$ is given by 
\begin{align*}
    f_{Y}(y)=\sum_{x \in \cX}P_{X}(x)f_{Y|X}(y|x)
\end{align*}
is also Reimann Integrable.\\
Further, the conditional probability $P_{X|Y}(x|y)$ is given by
\begin{align*}
    P_{X|Y}(x|y)=\frac{P_{X}(x)f_{Y|X}(y|x)}{f_{Y}(y)}
\end{align*}
The conditional entropy of $X$ given the random variable $Y$ is given by:
\begin{align*}
    H(X|Y)=\int_{\infty}^{\infty}f_{Y}(y)\Bigg(\sum_{x \in \cX}P_{X|Y}(x|y)\log\frac{1}{P_{X|Y}(x|y)}\Bigg)dy
\end{align*}
}%
We also need universal hash functions and strong randomness extractors for our commitment scheme; see~\cite{bloch} for detailed definitions. Finally, we borrow a classic result (cf.~\cite{shannonc,Gallager} and our appendix later) on existence of spherical codes for a given value of minimum distance.
\removed{
We also need the following definitions:
\begin{definition}[$\xi$-Universal hash functions~\cite{bloch}] 
	Let $\mathcal{H}$ be a class of functions from $\cX$ to $\cY$. $\mathcal{H}$ is said to be $\xi-$universal hash function, where $\xi\in\mathbb{N}$, if when $h\in\mathcal{H}$ is chosen uniformly at random, then $(h(x_1),h(x_2),...h(x_{\xi}))$ is uniformly distributed over $\cY^{\xi}$, $\forall x_1,x_2,...x_{\xi} \in \cX$.
\end{definition}
\begin{definition}[Strong randomness extractors~\cite{bloch}]
	 A probabilistic polynomial time function of the form \text{Ext}: $\{0,1\}^n \times \{0,1\}^d \to \{0,1\}^m$ is an  $(n,k,m,\epsilon$)-strong extractor if for every probability distribution $P_{Z}$ on $\cZ=\{0,1\}^n$, and $H_{\infty}(Z)\geq k$, for random variables $D$ (called 'seed') and $M$, distributed uniformly in $\{0,1\}^d$ and $\{0,1\}^m$ respectively, we have $||P_{Ext(Z;D),D}- P_{M,D}|| \leq \epsilon$.
\end{definition}
}%removed randomness extractor, UHF definitions.
%removed error correcting codes subsection

%
%\textcolor{red}{
%\begin{itemize}
%\item include Crepeau-style converse
%\item include theorem statement: impossibility with finite $P$, and also on finite rate bound under the Markov chain (using Crepeau)
%\end{itemize}
%}
\section{System Model and Problem Description}\label{sec:system:model}
%(henceforth, only referred to as \emph{elasticity}) where $E := \delta^2 - \gamma^2$, for $0<\gamma\leq\delta,$ and $\gamma,\delta\in\mathbb{R}^+.$\footnote{We require the strict inequality  $\gamma>0$ as otherwise a malicious Alice can force the \gunc\ \ to an AWGN channel with variance $0,$ i.e., a \emph{noiseless} channel. It is well known that a noiseless channel precludes commitment~\cite{blum}.}  
%
In our problem (ref Fig.~\ref{fig:main:setup}), two mutually distrustful parties, \emph{committer} Alice and \emph{receiver} Bob employ a \gunc \ to realize commitment over a uniformly random string $C\in[2^{nR}]$ available to Alice (we specify $R>0$ later). Alice and Bob have access to a one-way \gunc \ with two-sided elasticity $E$ at both Alice and Bob.
Separately, as is common in such cryptographic primitives, we also assume that Alice and Bob can interact over a two-way link that is noiseless and where the interaction is public and fully authenticates the transmitting party.
To commit to her random string $C,$ Alice uses the \gunc \ channel $n$ times and transmits over it her encrypted data $\bX=(X_1,X_2,\cdots, X_n)\in\mathbb{R}^n$; Bob receives a noisy version $\bY\in\mathbb{R}^n$ of Alice's transmission $\bX.$ Alice has an input power constraint $P>0,$ i.e., Alice can only transmit vectors $\bX\in\cS(P),$ where $\cS(P):=\{\bx\in\mathbb{R}^n: {\|\bx\|}_2\leq \sqrt{nP} \}$.  We allow private randomization at both Alice and Bob via their respective keys $K_A\in\mathcal{K_A}$ and $K_B\in\mathcal{K_B}.$ At any point in time, say time $i$, Alice and Bob can also exchange messages over the public, noiseless link prior to transmitting $X_i$; let $M$ denote the entire collection of messages exchanged over the noiseless link. We call $M$ the \emph{transcript} of the protocol. It is important to note that we assume that any point in time during the protocol, the transmissions of Alice and/or Bob can depend \emph{causally} on the information previously available to them. 

\begin{figure}[!ht]
  \begin{center}
    % \vspace{-5pt}
		\includegraphics[trim=3.5cm 6cm 5cm 0cm, scale=0.45]{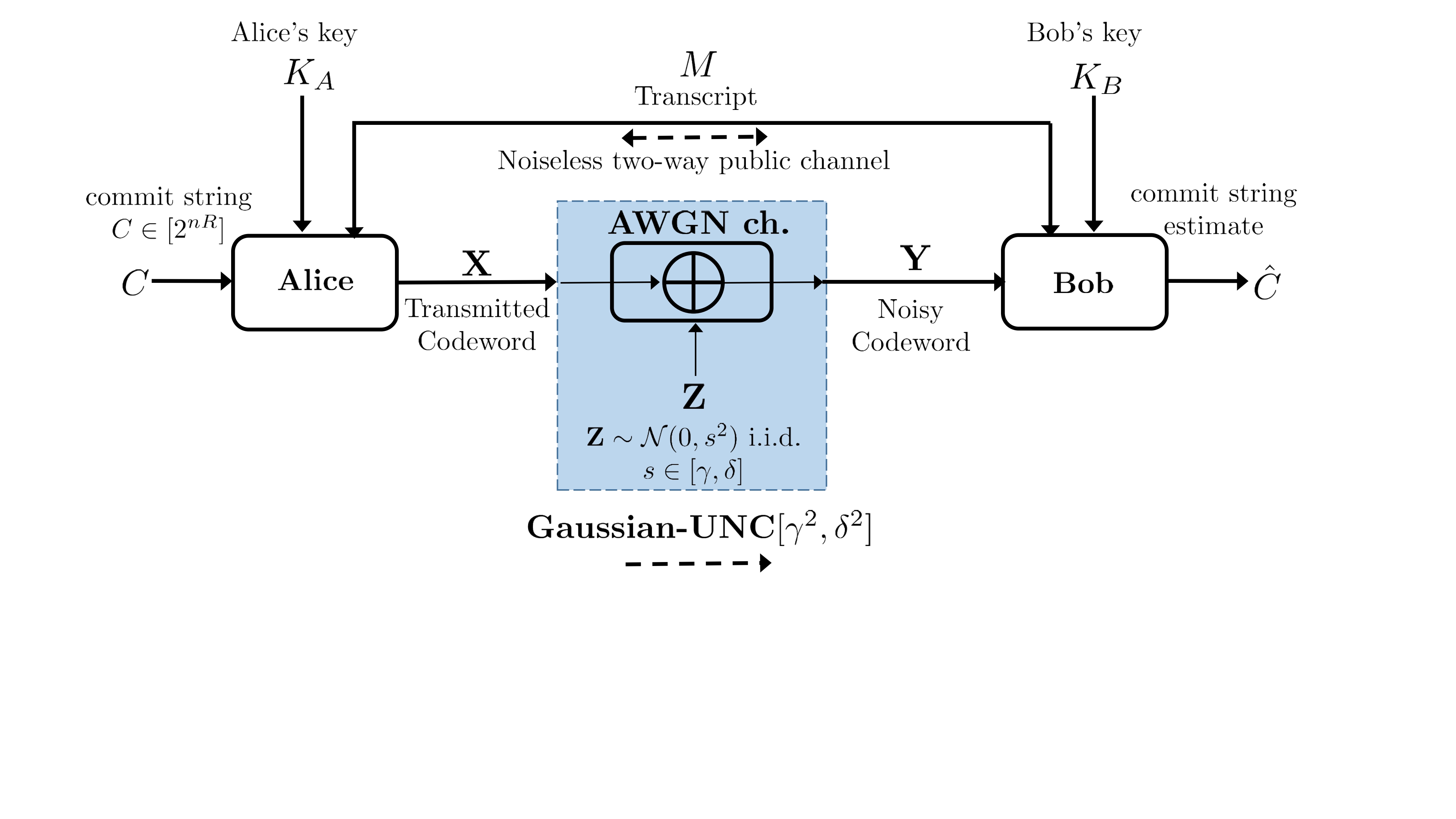}
    \caption{The problem setup: commitment over a  Gaussian UNC$[\gamma^{2},\delta^{2}]$}
    \label{fig:main:setup}
  \end{center}
\end{figure}
We formally define an $(n,R)$-commitment protocol.
\begin{definition}[Commitment protocol]\label{def:commitmentprotocol}
An $(n,R)-$commitment protocol $\mathscr{P}$ is a message-exchange procedure between the two parties to realize commitment over the random bit string  $C\in[2^{nR}].$ We call $R$ the rate of such a commitment protocol $\mathscr{P}.$ There are two phases to a  protocol $\mathscr{P}:$ \emph{commit phase} followed by the \emph{reveal phase.} \\
\indent (a) Commit phase: Given $C\in[2^{nR}]$, Alice transmits $\bX\in\cS(P)$ using the \gunc \ channel $n$ times. Bob receives $\bY.$ The two parties also exchange messages over the noiseless link. Let $M$ denote this transcript of protocol $\mathscr{P}$.\footnote{We assume that the transcript may contain arbitrarily large, though finite, messages.} Let $V_A$ and $V_B$ denote \emph{Alice's view} and \emph{Bob's view} respectively; these  `views' comprise all the random variables/vectors known to the two parties at the end of the commit phase.\\
\indent (b) Reveal phase: In this phase, Alice and Bob only communicate over the noiseless public link and do not use the \gunc \ \footnote{In some works, the noisy channel is leveraged to first realize the noiseless link and subsequently, the rate of the commitment protocol is suitably defined by amortizing the size of the commit string w.r.t. the total number of uses of the noisy channel. We do not study such a definition of rate in this work.}. Alice announces the commit string $\tilde{c}\in[2^{nR}]$ and $\td{\bX}\in\mathbb{R}^n.$ Bob then performs a test $T(\td{c},\td{\bX},V_B)$ and either accepts (by setting $T=1$) the commit string $\td{c}$ or reject it (by setting $T=0$).
\end{definition}

For a given $(n,R)-$commitment protocol, the following parameters are of key interest:
\begin{definition}[$\epsilon$-sound]\label{def:soundness}
A protocol $\mathscr{P}$ is $\epsilon-$sound if, for  honest Alice and  honest Bob,
\begin{align}
\mathbb{P}\left(T(C,\bX,V_B)\neq 1\right)\leq \epsilon.\notag
\end{align}
\end{definition}
\begin{definition}[$\epsilon$-concealing]\label{def:concealing}
A protocol $\mathscr{P}$ is $\epsilon-$concealing if, for honest Alice and under any strategy of Bob,
\begin{align}
I(C;V_B) \leq \epsilon.\notag
\end{align}
\end{definition}
\begin{definition}[$\epsilon$-binding]\label{def:binding}
A protocol $\mathscr{P}$ is $\epsilon-$binding if, for  honest Bob and under any strategy of Alice,
\begin{equation*}
\mathbb{P}\Big(T(\bar{c},\bar{\bx},V_B)=1 \quad\& \quad T(\hat{c},\hat{\bx},V_B)=1)\leq \epsilon\label{eq:binding}
\end{equation*}
for any two pairs $(\bar{c},\bar{\bx})$, $(\hat{c},\hat{\bx})$, $\bar{c}\neq \hat{c},$ and $\bar{\bx},\hat{\bx}\in\cS(P)$.
\end{definition}
A rate $R$ is said to be \emph{achievable} if for every $\epsilon>0$ there exists for every $n\in\mathbb{N}$ sufficiently large a protocol $\mathscr{P}$ such that  $\mathscr{P}$  is $\epsilon-$sound, $\epsilon-$concealing and $\epsilon-$binding. The supremum of all achievable rates is called the commitment capacity $\mathbb{C}$ of the \gunc \ . 

\section{Our Main Results}\label{sec:main:results}
%We state our results on the commitment capacity of the Gaussian UNC$[\gamma^2,\delta^2]$ under power constraint $P>0.$ 
We first present an impossibility result.
\begin{theorem}[Impossibility]\label{thm:impos}
For a Gaussian UNC$[\gamma^2,\delta^2],$ with unconstrained input ($P \to \infty$), the commitment capacity $\mathbb{C}=0$ if $\delta^2\geq 2\gamma^2$. 
\end{theorem}
%
%An equivalent interpretation of Theorem \ref{thm:impos} is that there doesn't exist an \emph{achievable} commitment protocol in the above regime.
The proof of this zero-rate converse is given in Section~\ref{sec:proofs}. 
%Our proof is inspired by \cite{damgaard1999possibility}.
\begin{remark}
    In the proof, we simulate the noisy channel functionality entirely via  noiseless interactions between the parties, and crucially leverage the impossibility of commitment (cf.~\cite{blum}) under such interactions. An interesting consequence of our proof, also stated as a corollary later (cf. Corollary~\ref{cor:smc}), is that similar impossibility can be argued for more general multi party functionalities, for e.g., oblivious transfer, over Gaussian UNCs for identical parameters.
% For our impossibility result, we assume that Alice has no input power restriction; in essence, commitment is seen to be impossible if the \emph{elasticity} $E=\delta^2-\gamma^2$ is \emph{large enough.} This result is similar in flavour to that over the binary UNCs, and is in fact, inspired by the one for binary UNCs by Damgard \emph{et al.}~\cite{damgaard1999possibility}. A key fact used in the converse is the classic result of impossibility of commitment over noiseless links (even when parties can privately randomize). Our proof of the converse continues to use this same approach. We first  construct another channel model (or functionality) called the \pgunc \ with identical parameters as in the \gunc \ (see Sec~\ref{sec:proofs:converse}). 
%We seek to establish a sequence of algorithmic reductions to establish that commitment is possible over a noiseless link; this being a  contradiction gives us the impossibility claim.
%In our proof, we show via a sequence of reductions that \pgunc\ can be simulated noiselessly and thus should preclude even single-bit commitment. We then conclude the argument by showing that impossibility over \pgunc\ , in turn, implies that commitment is impossible over \gunc.
\end{remark}
%
%\ajb{Move the Corollary in the proof of impossibility here.}
Having understood the impossibility regime (partly), we now flip the question and seek to explore positive-rate commitment schemes. Unlike the impossibility result where the inputs were unconstrained, we give an achievability rate as a function of input power constraint $P$. %Furthermore, we notice that there is a stark difference in our achievability results when $P>E$ and when $P\leq E.$ This motivates us to define the notion of \emph{signal-to-elasticity ratio (SER)}, where $SER:=P/E$ for a Gaussian UNC$[\gamma^2,\delta^2].$ 
We now state our first result. Due to space constraints, the proof details of all the following results are in the appendix.
\begin{theorem}\label{thm:possible:highSER}
For a Gaussian UNC$[\gamma^2,\delta^2]$ with $P>0$, positive-rate commitment  is possible if the following holds:
\begin{align}\label{eq:thresh}
\delta^2 < \left(1+\frac{P}{P+\gamma^2}\right)\gamma^2.
\end{align}
The commitment capacity $\mathbb{C}\geq \mathbb{C}_L$ where
%
%\begin{align}\label{eq:cap:highSER}
%$\mathbb{C}\geq \mathbb{C}_L.$
%\end{align}
%  

%
\begin{align}\label{eq:C:L}
 \mathbb{C}_L:=\frac{1}{2}\log\left(\frac{P}{\delta^2-\gamma^2}\right)-\frac{1}{2}\log\left(1+\frac{P}{\gamma^2}\right)
\end{align}
\end{theorem}
See overview in Section ~\ref{sec:proofs:achieve:highSER};  details are in the appendix.
\begin{remark}
The proof of this theorem uses a novel approach but crucially borrows ideas both from the protocol of Cr{\'e}peau \emph{et al.}~\cite{crepeau2020commitment} for binary UNCs as well as  that of Nascimento \emph{et al.}~\cite{nascimento-barros-t-it2008} for classic AWGN channels.  The `skeleton' of our protocol, where we use an error correcting code with certain minimum distance guarantee, is similar to that in the latter. However, to handle adversaries who may benefit from the channel elasticity available (unlike in classic AWGN channels which lack elasticity), we `robustify' our protocol by using  an appropriate hash function challenge mechanism inspired by the protocol in Cr{\'e}peau \emph{et al.}~\cite{crepeau2020commitment}.
% We present an overview of the scheme in Section~\ref{sec:}.
%While the soundness of our protocol follows from Chernoff bounding, the concealment and bindingness are more tricky. For concealment, we utilize a well known equivalence between the so-called \emph{bias-based} security and \emph{capacity-based} security; here, the leftover hash lemma~\cite{leftover} is crucially used. The bindingness analysis follows from the hash function challenges Bob offers to Alice. An important part in this analysis requires us to get a good bound on the number of `confusable' codewords that a cheating Alice may seek to reveal. Here we use results on spherical codes to get the appropriate bound on the cardinality of the `confusable' codewords.
\end{remark}
\begin{figure}[!ht]
  \begin{center}
    \vspace{5pt}
		\includegraphics[scale=0.5]{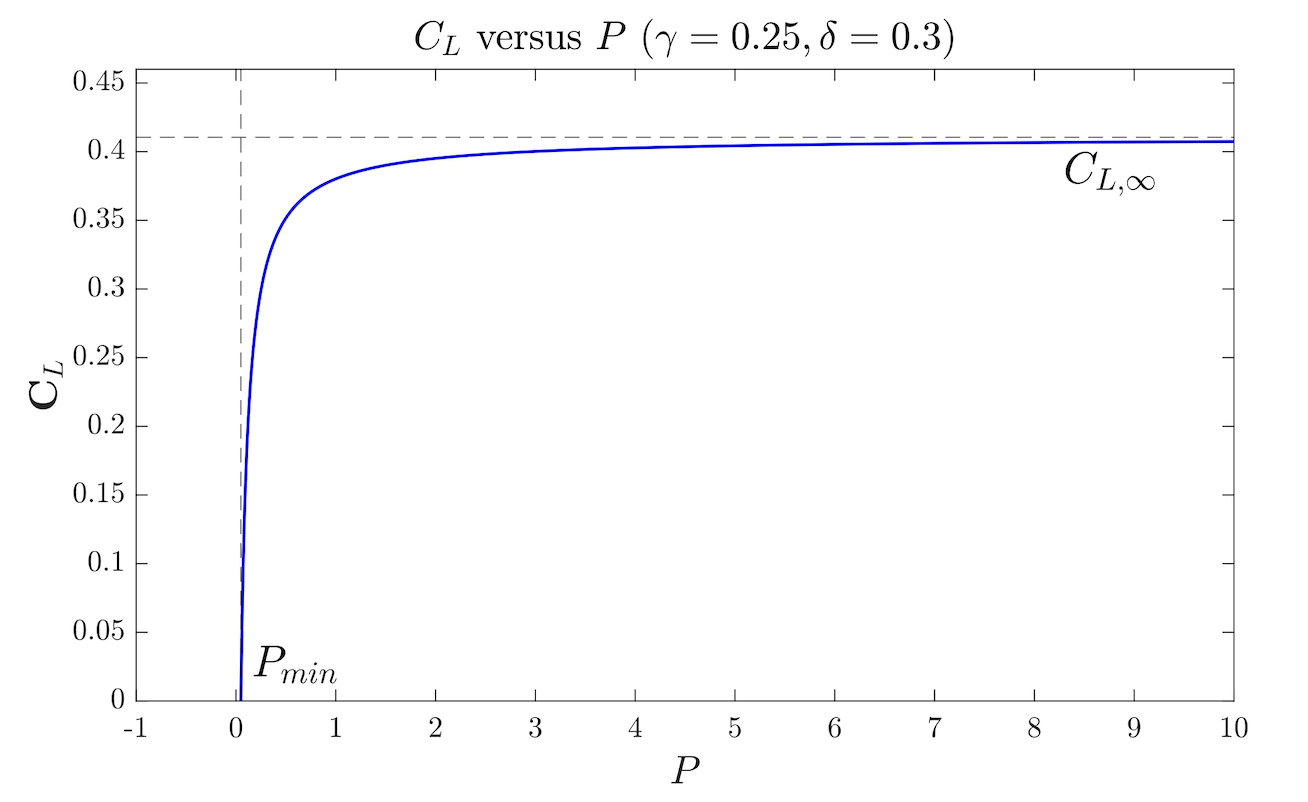}
    \caption{Plot of $\mathbb{C}_L$ versus $P$. For given $\gamma,\delta,$ $\mathbb{C}_L=0$ when $P\leq P_{\min}$. Also, $\mathbb{C}_L$ saturates to a finite value  $\mathbb{C}_{L,\infty}$ as $P$ increases.}
    \label{fig:lower:bound:P}
  \end{center}
\end{figure}
In Fig.~\ref{fig:lower:bound:P}, we present a representative plot of the lower bound $\mathbb{C}_L$ as a function of $P$ (for given values of $\gamma, \delta$). It is interesting to note that $\mathbb{C}_L$ becomes positive only for $P>P_{\min}:=\frac{\gsq (\dsq - \gsq)}{2\gsq - \dsq}$ and then saturates to $\mathbb{C}_{L,\infty}$ (as a concave function of $P$) as $P$ becomes large. Seen from another perspective, where we explore the variation of $\mathbb{C}_{L,\infty}$ w.r.t. the channel elasticity (cf. Fig.~\ref{fig:lower:bound:infty}), one notices that $\mathbb{C}_{L,\infty}$ is zero when the elasticity is sufficiently large, i.e., for fixed $\delta^2,$ when  $\gsq\leq\frac{\dsq}{2}$ (cf. Theorem~\ref{thm:impos}). 
\begin{figure}[!ht]
  \begin{center}
    \vspace{7pt}
		\includegraphics[scale=0.5]{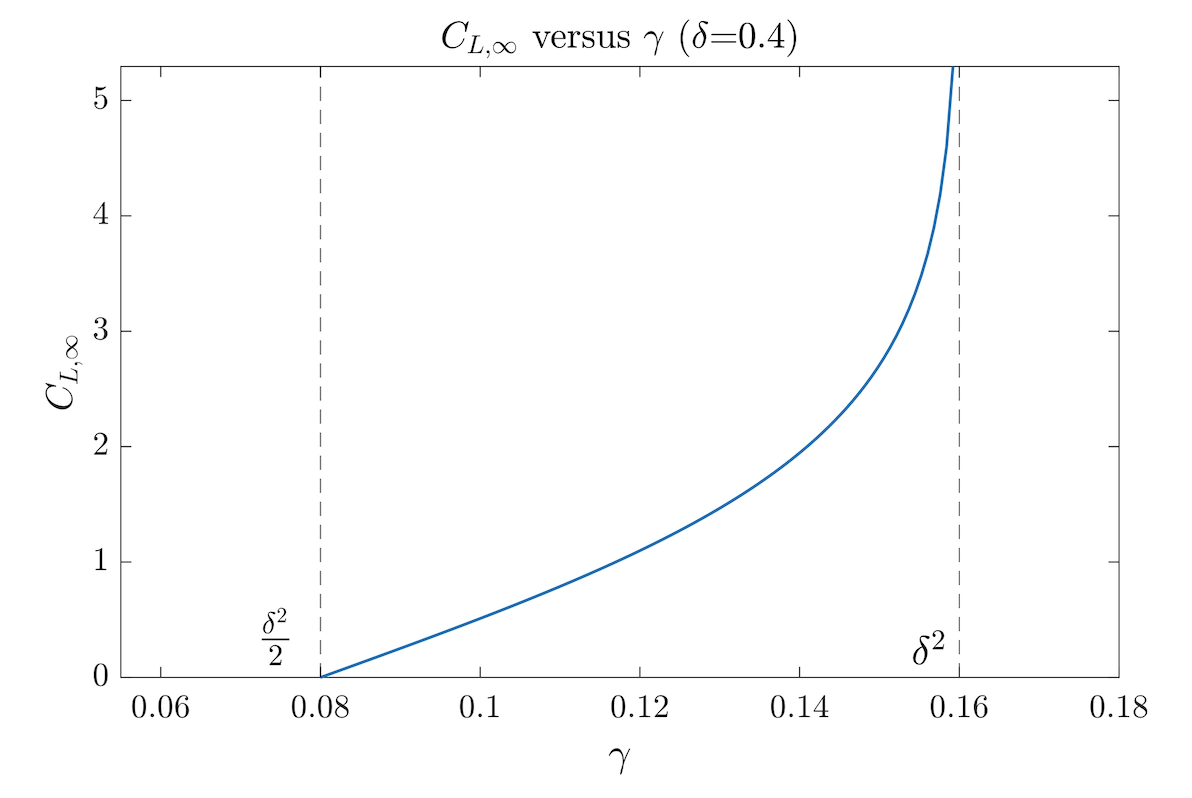}
    \caption{Plot of limiting value of lower bound $\mathbb{C}_L$ (as $P\rightarrow \infty$)  versus $\gamma^2$ for fixed $\delta^2$. The limiting lower bound $\mathbb{C}_{L,\infty}$, and hence, $\mathbb{C}$ tends to infinity as elasticity $E=\delta^2-\gamma^2\to 0$.}
    \label{fig:lower:bound:infty}
  \end{center}
\end{figure}
For fixed  $\gamma^2,\delta^2,$ 
%note that the positive-rate `possibility' threshold in~\eqref{eq:thresh}  depends only on the power constraint $P.$  Furthermore, 
the  positive-rate threshold in~\eqref{eq:thresh} `shifts' monotonically as a function of $P$ towards the impossibility bound  of $\dsq \geq 2 \gsq$ of Theorem~\ref{thm:impos}. In fact, in the limit $P\rightarrow \infty,$ the two bounds meet exactly, thereby allowing us to characterize precisely the positive commitment rate threshold when the input is unconstrained. We state this result in the following theorem.
\begin{theorem}\label{thm:achieve:inftyP}
Fix $\gamma^2,\delta^2<\infty$ and let $P\rightarrow \infty.$ Then commitment is possible if and only if $\delta^2< 2\gamma^2.$
\end{theorem}
The proof of this result simply follows from Theorem~\ref{thm:impos} and by taking the limit $P\rightarrow \infty$ for the threshold in~\eqref{eq:thresh} (for fixed $\gamma^2,\delta^2$) from Theorem \ref{thm:possible:highSER}. %We skip the details due to space constraints.
An important take away from the achievability results is that the possible `finiteness' of commitment capacity (recall that we only present a lower bound when $P>0$ is finite) of Gaussian UNCs is owing to the underlying channel elasticity $E$. This fact is particularly stark when, for fixed $P>0,$ one allows the channel elasticity  to vanish, i.e., $E \to 0.$ In such a case, the commitment capacity lower bound (see Theorem~\ref{thm:possible:highSER}), and hence, the commitment capacity, becomes  infinite, which is exactly what is known for classical AWGN channels (where $\gamma^2=\delta^2$)  that exhibit $E=0.$ We capture this alternate perspective on the infinite commitment rate of classical AWGN channels in the following corollary.
\begin{corollary}
For a fixed $P>0,$ the commitment capacity of a Gaussian UNC$[\gamma^2,\delta^2]$ with channel elasticity $E\to0,$ i.e., any classic AWGN channel where $\gamma^2 \to \delta^2$, is infinite, irrespective of the fixed noise variance $\delta^2$. 
\end{corollary}

\section{Proofs}\label{sec:proofs}
\subsection{Converse: Proof of Theorem~\ref{thm:impos}}
A key fact we use in our proof and stated later in Claim \ref{clm:impossible:noiseless} is that commitment is impossible over noiseless channels. 
Hence, to argue to argue that commitment is impossible over \guncw s it suffices to simulate every use of the channel noiselessly. We do this in Lemma \ref{lem:reduction:noiseless}, the proof of which forms the principal portion of this section and is inspired by techniques used in \cite{damgaard1999possibility} for binary UNCs. The proof heavily relies on the conditions $\dsq \geq 2\gsq$ and $P \to \infty$ i.e., the channel is unconstrained. Throughout the section we assume that they hold.
Consider a scheme $\scrP$ that makes use of \guncw\gsqdsq\ channels to achieve commitment. We start by formalizing the notion of a ``use'' of a \guncw\gsqdsq\ by defining a two party interactive protocol \GUNC\gsqdsq\ in Algorithm 1 which $\scrP$ calls as a subroutine. Correspondingly we design a  \SGUNC\gsqdsq\ in Algorithm 2 that uses purely noiseless operations to simulate \GUNC\gsqdsq. The inputs to these protocols are the private views of Alice and Bob when the subroutine is invoked by $\scrP$. We denote them by $U_A$ and $U_B$ respectively.\footnote{not to be confused with the views in Definition \ref{def:commitmentprotocol}} 
We indicate by $x \in \bbR$ an input Alice is prescribed by $\scrP$ to send over the noisy channel. Although her view $U_A$ includes $x$, we explicitly denote $(x,U_A)$ as Alice's inputs to the protocol for readability purposes.
At the end of the protocol they both output prescribed values into their respective private views and forego any additional knowledge they may have gained through the execution of the subroutine. While an honest party sticks to the prescriptions, a cheating party may deviate whenever possible.\\
\begin{algorithm}\label{alg:gunc}
    \caption{\GUNC\gsqdsq$\left((x, U_A); U_B\right)$}

    \hspace*{\algorithmicindent} \textbf{Inputs:} Alice: $(x, U_A) = U_A $ \quad Bob: $ U_B $ 
    \begin{algorithmic}[1]
        \State The oracle decides $\theta^2 = \cO(\cdot)$ from \gsqdsq
        \State Alice sends $\tdx = x$ over the channel
        \State Bob receives $Y = \tdx + Z$ s.t. $Z \sim \cN(0,\tsq)$
        \State Alice, Bob output $x$, $Y$ respectively to their views
    \end{algorithmic}
\end{algorithm}\\
All \GUNC\footnote{We drop the suffix \gsqdsq\ from here on for reading convenience.} does is to execute a single run of the \guncw\ channel. The arbitrary parameter instantiation of the channel is modeled by an arbitrary fixed\footnote{It is important from the definition of \guncw\ that $\theta^2$ remain fixed for multiple runs of \GUNC.} oracle function $\cO$ who outputs noise parameter $\theta^2$ in \gsqdsq. \SGUNC\ attempts to replicate the AWGN $Z$ with two local random variables $Z_1$ and $Z_2$.\\
\begin{algorithm}\label{alg:sgunc}
    \caption{\SGUNC\gsqdsq$\left( (x,U_A); U_B\right)$ }
    \hspace*{\algorithmicindent} \textbf{Inputs:} Alice: $(x, U_A) = U_A$ \quad Bob: $ U_B$ 

    \begin{algorithmic}[1]
        \State Alice computes $W = \tdx + Z_1$ where $Z_1 \sim \cN(0,\gsq)$ and sends $W$ to Bob 
        \State Bob computes $Y = W + Z_2$ where $Z_2 \sim \cN(0,\gsq)$
        \State Alice, Bob output $x$, $Y$ respectively to their views.
    \end{algorithmic}
\end{algorithm}\\
To show that \SGUNC\  perfectly simulates \GUNC\ we compare the joint output distributions of Alice and Bob in both the protocols. Note that the output depends on whether the parties remain honest or employ a cheating strategy, which we model with arbitrary random functions at each step of the protocol. These are functions of all the information the respective party has at that instant. 
We seek to show that for all cheating attacks $\cA$ (by Alice) and $\cB$ (by Bob) on  \SGUNC, there exist corresponding cheating attacks $\tdcA$, $\tdcB$  on  \GUNC\ that result in the same output distributions.
It suffices to show that this holds for a fixed oracle $\cO$.
In doing so we argue that if there is a cheating attack on \SGUNC\ that precludes commitment (which we show is true in Claim \ref{clm:impossible:noiseless}), there is also a corresponding cheating strategy that can preclude commitment in \GUNC. 
On the other hand, any general scheme $\cP$ that uses \GUNC\ as a subroutine and is robust against cheating attacks, can alternatively use \SGUNC\ and assuredly remain secure. 

For completeness, we give a comparison of the output distributions of both the protocols over different cases.
\\

\noindent \textbf{(a) Alice and Bob are honest:} 
We set $\cO(\cdot) = 2 \gsq$ which is in \gsqdsq. This results in a joint output $(x;x+Z)$ for  \GUNC\ where $Z$ $\sim \cN (0,2\gamma^2)$. The joint output of  \SGUNC\ is $(x;x + Z_1 +Z_2)$ where $Z_1, Z_2 \sim \cN(0,\gamma^2)$. Since $Z_1$ and $Z_2$ are independent, both the output distributions match.

\noindent \textbf{(b) Alice cheats and Bob is honest:}
Let $\cA_i$ be the attack used by Alice in the $i$th step of \SGUNC. Instead of following the prescribed steps, she now sends $W = \cA_1(U_A)$ to Bob in step one, and outputs $\cA_3(W,U_A)$ in step three. Bob (being honest) outputs $Y = W + Z_2$. \SGUNC's output can therefore be written as
$\left(\cA_3\left( W , U_A\right) ; W + Z_2 \right)$ s.t.  $W = \cA_1(U_A)$, $Z_2 \sim \cN(0,\gsq).$ 

We show one of Alice's attacks that induces the same joint output in \GUNC. She sets $\theta^2$ according to $\tdcA_1(U_A) = \gamma$ in step one, chooses $\tdX$ using\footnote{Capital $\tdX$ used to account for random attacks.} $\tdcA_2(\theta^2, U_A) = \cA_1(U_A)$ to send over the noisy channel in step two and outputs according to $\tdcA_4(\tdX,\theta^2,U_A) = \cA_3(\tdX,U_A)$ in step four. Bob outputs $Y = \tdX + Z$. The resulting joint output is 
$\left( \tdcA_4(\tdX ,\theta^2 ,U_A) ; \tdX + Z  \right)$ s.t. $\tdX = \tdcA_2 (\theta^2,U_A)$, $\theta^2 = \tdcA_1(U_A)$.
This evaluates to
$\left( \cA_3(\tdX,U_A);  \tdX + Z \right)$ s.t. $\tdX = \cA_1(U_A)$, $\theta^2 = \gsq$, $Z \sim \cN(0,\tsq).$ Both the distributions match.

\noindent \textbf{(c) Alice is honest and Bob cheats:}
In a similar fashion let $\cB_i$ be the strategy used by Bob in the $i$th step of \SGUNC. An honest Alice computes $W = x + Z_1$. She hands over $W$ to Bob and outputs $x$. The protocol allows Bob to cheat only in the third step by outputting $\cB_3(W,U_B)$. The joint output is 
$\left( x; \cB_3(x + Z_1,U_B) \right)$ s.t. $Z_1 \sim \cN(0,\gsq).$

In the \GUNC\ protocol Alice sends $\tdx = x$ and outputs $x$ as prescribed. We design an strategy for Bob where he sets $\theta^2$ to $\tdcB_1(U_B) := \gamma$ in step one and outputs $\cB_3(Y, \theta^2, U_B) := \cB_3(Y,U_B)$ in the fourth step. Here $Y = x + Z$. The joint output is 
$\left( x; \tdcB_3(x + Z, \theta^2, U_B) \right)$
s.t. $\theta^2 = \tdcB_1(U_B)$.
This evaluates to
$\left( x; \cB_3(x + Z_1, U_B) \right)$
s.t. $\theta^2 = \gsq$, $Z \sim \cN(0,\theta^2)$.
Both the distributions match.
 It is worthwhile to note here that our choice of oracle function $\cO(\cdot) = 2\gsq$ for $\theta^2$ was only possible because $\dsq \geq 2\gsq$.
We can now make the following lemma\footnote{It is not clear how to extend our simulation to channels with a finite input power constraint. We believe it requires some non-trivial techniques and leave it as an open problem.}.
\begin{lemma}\label{lem:reduction:noiseless}
    Every use of a \guncw\gsqdsq\ with unconstrained inputs and $\dsq \geq 2 \gsq$ by a commitment scheme can be simulated by one-way noiseless transmissions.
\end{lemma}

All that now remains is to see if information-theoretically secure commitment is possible without using noise. We state this well known result without proof due to space constraints.
\begin{claim}\label{clm:impossible:noiseless}
	No $\e_1-$sound, $\e_2-$concealing, $\e_3-$binding $k-$bit commitment scheme is possible over noiseless interactions for 
	\begin{equation*}
		\e_2 < k(1-\e_1 -  2^k \e_3) - 2\sqrt{\e_1 + 2^k \e_3}.
	\end{equation*}
\end{claim}

From Lemma \ref{lem:reduction:noiseless} and Claim \ref{clm:impossible:noiseless} we can conclude that for \guncw\gsqdsq\ with $\dsq\geq2\gsq$, there exists a finite $\epsilon = 2^{-3k}$ for which there is no scheme $\scrP$ that is \emph{$\epsilon$-sound, $\epsilon$-concealing and $\epsilon$-binding} i.e., \emph{achievable}.

It is easy to see that the reduction in Lemma \ref{lem:reduction:noiseless} can be extended for any multi party computation functionality as long as there are no security guarantees when both the sender and receiver are cheating. This gives us an impossibility result for a number of functionalities like oblivious transfer which  are known to be impossible over noiseless communication. 
\begin{corollary}\label{cor:smc}
    Consider a multi-party functionality that (i) gives no security guarantees for a pair of simultaneously cheating parties and (ii) is known to be impossible using noiseless interactions. The functionality is also impossible when the pair is allowed to use a \guncw\gsqdsq\ channel with $\dsq\geq2\gsq$ and has no input constraints.
\end{corollary}

\subsection{Achievability: Proof of Theorem~\ref{thm:possible:highSER}}\label{sec:proofs:achieve:highSER}
%the detailed protocol and its analysis is in the appendix. We leverage  ideas from~\cite{damgaard1999possibility,nascimento-barros-t-it2008,crepeau2020commitment}; however, our specific protocol for the Gaussian UNCs is novel to the best of our knowledge. 

\textbf{Overview:}
In the commit phase, Alice first generates a random bit string $U^m\in\{0,1\}^m$ towards committing a uniformly random string $C\in[2^{nR}].$
Alice then uses an error correcting code, say $\cC=(\psi,\phi)$ ($\cC\subseteq \mathbb{R}^n$ and is known to both parties), to encode this bit string $U^m$  to codeword $\bX=\psi(U^m)$ and transmits $\bX$ over the \gunc \  to Bob. Our  error correcting code $\cC$ (as in~\cite{shannonc,Gallager}) is a spherical code comprising \emph{equi-normed}  codewords (where all codewords reside on the surface of a $n$-dimensional Euclidean ball). Bob receives a noisy version $\bY$ of the transmitted codeword $\bX$.
We choose the rate $\bar{R}(\cC)$ of the error correcting code $\cC$ `sufficiently large' (see details in appendix); this ensures that  upon receiving a noisy observation $\bY$ of the transmitted codeword $\bX,$ Bob decodes a `large' list $\cL(\bY)\subseteq \cC$ of  codewords which are `typical' with respect to (w.r.t.) the observation $\bY$  (here typicality is w.r.t. the underlying \guncw\gsqdsq). 

Recall however that a cheating Alice can privately change the noise variance in the \gunc \ . Such an action can `enlarge' her set of `spoofing codewords' that she can present, if dishonest, in the reveal phase. To restrict Alice's potential dishonest behaviour, our protocol employs the classic hash-challenge approach (inspired by~\cite{damgaard1999possibility}). In particular, Bob initiates a \emph{two-round} hash challenge with Alice\footnote{We need two rounds of hash challenge to circumvent a non-trivial rate loss that arises in the single hash challenge due to the \emph{birthday paradox}; see~\cite{crepeau2020commitment, damgaard1999possibility} where it is discussed.} which essentially \emph{binds} Alice to her choice of $U^m$ (remember $U^m$ has a $1-1$ mapping with $\bX$ via the codebook $\cC$)   in the commit phase thereby ensuring Bob's test $T$ can detect any cheating attempt by Alice during the reveal phase. 
Essentially, the first hash challenge reduces the number of strings that Alice can use to confuse Bob in the reveal phase from exponential to polynomial in block-length $n$; the second hash challenge further brings down the number of such bit strings to one (this precludes the possibility of Bob being confused between two different bit string, say $U_1^m,U_2^m\in\{0,1\}^m,$ thereby ensuring the binding guarantee). 
Additionally, a strong randomness extractor is used by the committer Alice which extracts a secret key (the privacy amplification lemma~\cite{glh} allows us to quantify the size of this key) from $\bX$. This secret key is then XOR-ed with the commit string $C$ to realize a \emph{one-time pad} (OTP) scheme, which ensures concealment of the committed string against a malicious Bob in the commit phase.

Finally, in the reveal phase, Alice reveals the bit string $\tilde{u}^{m}$ to Bob over the noiseless link. Bob recovers $\td{\bx}=\phi(\tilde{u}^{m}),$ and then verifies it through a series of tests based on typicality, hash challenges, and the OTP-based randomness extractor.
%Owing to space constraints, the protocol has been described in the Appendix below (Appendix~\ref{subsec:app:ach1}) followed by the detailed analysis of the security guarantees- \emph{soundness, concealment} and \emph{bindingness}  (Appendix~\ref{subsec:app:ach2},~\ref{subsec:app:ach3}).

\clearpage
\balance
\bibliographystyle{IEEEtran}
\bibliography{IEEEabrv,References,refs_AKY}
%\removed{% removed appendices
\newpage
\appendix
\subsection{Our Achievability Protocol}\label{subsec:app:ach1} 
We now present our commitment protocol for  $P>P_{\min}$.
Alice and Bob fix a spherical error correcting code $\cC\subseteq \mathbb{R}^n$ (as in~\cite{shannonc,Gallager}) comprising an encoder $\psi:\{0,1\}^m\rightarrow \mathbb{R}^n$ and decoder $\phi:\mathbb{R}^n\rightarrow \{0,1\}^m\bigcup \{0\}$ with rate $\bar{R}=\frac{m}{n}:=\frac{1}{2}\log(\frac{1}{1-\left(1-\frac{\hat{d}}{2}\right)^2})-\td{\beta}$ such that $\|\bx\|=nP$, $\forall \bx \in \cC$ and $d_{\min}(\cC)=n\hat{d}^{2}P$ is the minimum distance of the code $\cC.$ 
The commitment rate of the protocol is 
\begin{align}\label{eq:commit:rate}
R=\frac{1}{2}\log\left(\frac{P}{E}\right) -\frac{1}{2}\log\left(1+\frac{P}{\gamma^2}  \right)-\beta_3.
\end{align}

Let $\mathcal{G}_1:=\{g_1:\{0,1\}^m \rightarrow \{0,1\}^{n(\bar{R}+\frac{1}{2}\log(\frac{E}{P})+ \beta_1)}\}$ be a $3n\bar{R}$-universal hash family, where 
$E:=\delta^2-\gamma^2$ and 
$\beta_1>0$ is a small enough constant. We choose sufficiently large $\bar{R}>\frac{1}{2}\log\left(\frac{P}{E} \right)$ so that $\cG_1$ is meaningfully defined.
Let  $\mathcal{G}_2:=\{g_2:\{0,1\}^m \rightarrow \{0,1\}^{n\beta_2}\}$ be a $2-$universal hash family, where $\beta_2>0$ is a small enough constant.
Let $\mathcal{E}:=\{\text{ext}:\{0,1\}^m \rightarrow \{0,1\}^{nR}\}$ be a $2-$universal hash family, where $\beta_3>0$ is chosen such that $\beta_3 > \beta_1 + \beta_2$.\footnote{Note that  $R$ can be made arbitrarily close to $\mathbb{C}_L.$}\\
Here are the commit and reveal phases of our protocol $\mathscr{P}$:\\

\noindent $\bullet$ \textbf{\underline{Commit Phase:}}  \\

\noindent Alice seeks to commit to string $C\in[2^{nR}]$ and proceeds as follows:
\\

\noindent (C1). Given $C$, Alice first generates $U^m=(U_1, U_2, \cdots, U_m)\sim \text{Bernoulli}(1/2)$ independent and identically distributed (i.i.d.) bits. 
\\
\noindent (C2). Using code $\cC=(\psi,\phi)$,  Alice picks the codeword $\bX=\psi(U^m)$ and sends it over the \gunc \. Let Bob receive $\bY$ over the noisy channel. 
\\
\noindent (C3). Bob creates a list $\cL(\by)$ of codewords in $\cC$ given by:\footnote{Here the parameter $\alpha_1>0$ is chosen appropriately small.}
\begin{IEEEeqnarray*}{rCl}
\cL(\by):=\{\bx\in \cC: n(\gamma^2 -\alpha_1) \leq ({\|\bx-\by\|_2})^2 \leq n(\delta^2 +\alpha_1) \}.
\end{IEEEeqnarray*}
\noindent (C4). Bob now initiates the two rounds of hash challenges for Alice. Bob first chooses the hash function $G_1\sim \text{Unif}\left(\mathcal{G}_1\right)$. Bob sends the description of $G_1$ to Alice over the two-way noiseless link.
\\
\noindent (C5).  Using $G_1,$ Alice computes the hash  $G_1(U^m)$ and sends the hash value, say $\bar{g}_1,$ to Bob over the noiseless link.
\\
\noindent (C6). Next, Bob initiates the second round of hash exchange by choosing another hash function $G_2\sim\text{Unif}\left(\mathcal{G}_2\right)$, and  sends the description of $G_2$ to Alice over the noiseless link.
\\
\noindent (C7). Once again, Alice locally computes the hash value $G_2(U^m)$ and sends the hash value, say $\bar{g}_2$, to Bob over the noiseless link.
\\
\noindent (C8). Alice now chooses an extractor function $\texttt{Ext}\sim\text{Unif}\left(\mathcal{E}\right)$ and sends\footnote{The operator $\oplus$ here denotes component-wise XOR.} the one-time pad (OTP) $Q = C \oplus \texttt{Ext}(U^m)$ along with the exact choice of the function $\texttt{Ext}$ to Bob over the noiseless link.\\

\noindent $\bullet$ \textbf{\underline{Reveal Phase:}} \\

\noindent The following operations comprise the reveal phase:
\\

\noindent (R1). Alice announces  $(\td{c},\td{u}^m)$ to Bob over the noiseless link.
\\
\noindent (R2). Bob determines the codeword $\td{\bx}=\td{\bx}(\td{u}^m)=\psi(\td{u}^m).$ 
\\
%\noindent (R3). Bob checks if $\td{\bx}\in\cL(\by),$ where $\by$ is Bob' observation over the \gunc at the end of the commit phase. If $\td{\bx}\not\in\cL(\by)$, Bob aborts.
\noindent (R3). Bob accepts $\td{c}$ if all the following four conditions are simultaneously satisfied: 
\begin{enumerate}[(i)]
\item $\td{\bx}\in\cL(\by)$, where $\by$ is Bob's observation over the noisy channel at the end of the commit phase.
\item $g_1(\td{u}^m)=\bar{g}_1$, 
\item  $g_2(\td{u}^m)=\bar{g}_2$,
\item $\tilde{c}=q~\oplus \texttt{ext}(\td{u}^m)$. 
\end{enumerate}
Else, he rejects $\td{c}$ and outputs `0'.
\subsection{Positivity of rate $R$ of our protocol $\mathscr{P}$:}\label{subsec:app:ach2} 
We first show that the rate $R>0$ when  $P>P_{\min}$, i.e., $(\delta^2-\gamma^2)< \frac{P\gamma^2}{P+\gamma^2}.$     Toward proving rate positivity, let us assume that $(\delta^2-\gamma^2)= \frac{P\gamma^2}{P+\gamma^2} -\eta,$ for some $\eta>0.$ Recall that the rate of the commitment protocol is 
\begin{align}
R&\stackrel{}{=}\frac{1}{2}\log\left(\frac{P}{E}\right) -\frac{1}{2} \log\left(1+\frac{P}{\gamma^2}  \right)-\beta_3\\
&\stackrel{}{=}\frac{1}{2}\log\left(\frac{P}{\delta^2-\gamma^2}\right) -\frac{1}{2} \log\left(1+\frac{P}{\gamma^2}  \right)-\beta_3\\
&\stackrel{}{=}\frac{1}{2}\log\left(\frac{P}{\delta^2-\gamma^2}\right) -\frac{1}{2} \log\left(\frac{P+\gamma^2}{\gamma^2}  \right)-\beta_3\\
&\stackrel{}{=}\frac{1}{2}\log\left(\frac{\frac{P\gamma^2}{(P+\gamma^2)}}{\delta^2-\gamma^2}  \right)-\beta_3
\end{align}
Given $\eta>0,$ for $\beta_3=\beta_3(\eta)>0$ small enough, it follows that $R>0.$

\subsection{Security Analysis}\label{subsec:app:ach3}
%We now analyse and prove the security guarantees in detail for the above defined $(n,R)$-commitment protocol:
\vspace{2mm}
\noindent\emph{\textbf{\underline{(i). $\e-$  soundness:}}}\\
    
%For our protocol to be $\epsilon$-sound, we essentially need to show that when both parties are honest, Bob accepts $\td{C}=C$ with high probability (w.h.p.). 
Since Alice and Bob are honest, it follows directly that it is sufficient to show that  $\bbP\left(\bX\not\in \cL(\bY)\right)\leq \epsilon$ for $n$ large enough. This is because, conditioned on the event $\{\bX\in\cL(\bY)\},$  the rest of the three conditions are deterministically true when both parties are honest. The classic Chernoff bound gives us the necessary bound.\\
%, viz., (a) $g_1(\td{u}^m)=\bar{g}_1$, (b) $g_2(\td{u}^m)=\bar{g}_2$, and (c) $\tilde{c}=q~\oplus \texttt{ext}(\td{u}^m)$ \emph{deterministically} hold true when Alice and Bob are both honest. The proof of the fact that $\bbP\left(\bX\not\in \cL(\bY)\right)\leq \epsilon$ for $n$ sufficiently large follows from classic Chernoff bound for additive Gaussian channels (including the unfair noisy version).\\ 

\noindent\emph{\textbf{\underline{(ii). $\epsilon-$concealing:}}}\\

Our approach uses the classic left-over hash lemma to show that the 2-universal hash function can be used as a strong randomness extractor to extract the `residual' randomness in the transmitted codeword $\bX$ and hence $U^{m}$ (recall that $\bX=\psi(U^m)$). It is well known that a positive rate commitment protocol is $\epsilon-$concealing for all $\epsilon>0$ for sufficiently large block length $n$, if it satisfies the \emph{capacity-based secrecy} notion (cf.~\cite[Def.~3.2]{damgard1998statistical}) and vice versa. We use a well established relation between \emph{capacity-based secrecy} and the \emph{bias-based secrecy} (cf.~\cite[Th.~4.1]{damgard1998statistical}) to prove that our protocol is $\epsilon$-concealing. 

  We first prove that our protocol satisfies bias-based secrecy by essentially proving the perfect secrecy of the key $\texttt{Ext}(U^m)$; we crucially use the  \emph{leftover hash} lemma. Several versions of this lemma exists (cf.~\cite{impagliazzo1989pseudo,glh} ); we use the following:
\begin{lemma}\label{lem:glh}
Let $\mathcal{G}=\{g:\{0,1\}^n\rightarrow \{0,1\}^l\}$ be a family of universal hash functions. Then, for any hash function $G$ chosen  uniformly at random from $\mathcal{G}$, and $W$
\begin{align*}
    \text{SD}(P_{G(W),G},P_{U_l,G}) \notag&\leq \frac{1}{2}\sqrt{2^{ -H_{\infty}(W)} 2{^l}}
\end{align*}
where $U_l\sim\text{Unif}\left(\{0,1\}^l\right).$
\end{lemma}
		%here we lower bound $H_{\infty}(\bX|\bY,G_1(\bX),G_1,G_2(\bX),G_2)$ and then  crucially use the generalized leftover hash lemma (cf.~\cite{dodis2004fuzzy}).
%
%    Then, we will use~\cite[Th.~4.1]{damgard1998statistical} to conclude that our scheme achieves the capacity-based secrecy, and show that $I(C;V_B)$ is exponentially decreasing for sufficiently large $n$. This will establish that our protocol is $\epsilon$-concealing, where $\epsilon>0$ is exponentially decaying in blocklength $n$.
%
\removed{ %simple itemize plan
\begin{itemize}
    \item $H_{\infty}(U^m|\bY,G_1(U^m),G_1,G_2(U^m),G_2)=\lim_{\Delta\rightarrow 0} H_{\infty}(U^m|\bY^{\Delta},G_1(U^m),G_1,G_2(U^m),G_2)$
    \item $H_{\infty}(U^m|\bY^{\Delta},G_1(U^m),G_1,G_2(U^m),G_2)=\lim_{\epsilon_1 \rightarrow 0} H_{\infty}^{\epsilon_1}(U^m|\bY^{\Delta},G_1(U^m),G_1,G_2(U^m),G_2) $
    \item Now lower bound the following quantity: $H_{\infty}^{\epsilon_1}(U^m|\bY^{\Delta},G_1(U^m),G_1,G_2(U^m),G_2)$
    \item now proceed with the lower bound and split to get $H_{\infty}^{\epsilon_1}(U^m|\bY^{\Delta}$ and $H_0$
    \item Use the fact that $H_{\infty}^{\epsilon_1}(A^n|B^n)\geq H(A|B)-n\delta$ for $n$ to conclude that  
\end{itemize}
}%removed

We seek to lower bound $H_{\infty}(U^m).$ Toward this, we analyse the conditional min-entropy of $U^m$ conditioned on  $V_B$ after the hash challenge (this quantity lower bounds the min-entropy of interest). However, owing to the continuous alphabet of Bob's observation $\bY,$ we need to take a `discretization approach' to first ``quantize'' the channel output, say via $\bY^{\Delta},$ and then calculate the conditional min-entropy over  $\bY^{\Delta}$. This is important since min-entropy and conditional min-entropy (as well as their \textit{smooth} versions) do not posses the properties we seek under continuous variables. \\
Our treatment is inspired from~\cite{nascimento-barros-t-it2008,cover-thomas}. Let $Y$ be a continuous random variable in $\mathbb{R}$ and $\Delta>0$ be some constant. Then, from the mean value theorem, there exists a $y_k$ such that
\begin{align*}
    f_{Y}(y_k)= \frac{1}{\Delta}\int_{\Delta k}^{\Delta(k+1)}f_{Y}(y)dy
\end{align*}
Let $X\in\cX$. Then, the conditional distribution:
\begin{align*}
    f_{Y|X}(y_k|x)=\frac{1}{\Delta}\int_{\Delta k}^{\Delta(k+1)}f_{Y|X}(y|x)dy
\end{align*}
Let $Y^{\Delta}$ represent the quantized version of the continuous random variable $Y$, which takes value $y_k$ for every $Y \in [\Delta k, \Delta(k+1)]$, with probability $P_{Y^{\Delta}}(y_k)=f_{Y}(y_k)\Delta$. Then, %Further, the joint probability distribution of the random variables $XY^{\Delta}$ is given as:
\begin{align*}
    P_{XY^{\Delta}}(x,y_k)=P_{X}(x)P_{Y^{\Delta}|X}(y_k|x)=P_{X}(x)f_{Y|X}(y_k|x)\Delta
\end{align*}
The quantized version of the conditional min-entropy is:
\begin{align*}
    H_{\infty}(X|Y^{\Delta})&=\inf_{x,y_k}(-\log(P_{X|Y^{\Delta}}(x|y_k)))\notag \\
    &= \inf_{x,y_k}\log\Bigg(\frac{f_{Y}(y_k)\Delta}{P_{X}(x)f_{Y|X}(y_k|x)\Delta}\Bigg)
\end{align*}
%
%ajb{AKY introduces quantization quantities here if possible.}
For $U^{m}$, note that for quantization via $\Delta>0,$ we have
\begin{align*}
    H_{\infty}(U^m|&\bY,G_1(U^m),G_1,G_2(U^m),G_2)\\&=\lim_{\Delta\rightarrow 0} H_{\infty}(U^m|\bY^{\Delta},G_1(U^m),G_1,G_2(U^m),G_2)
\end{align*}
where $\bY^{\Delta}$ is discrete and a quantized version of $\bY$.

Furthermore, from the definition of smooth-min-entropy~\cite{chain1}, we know that 
\begin{align*}
    H_{\infty}(U^m|&\bY^{\Delta},G_1(U^m),G_1,G_2(U^m),G_2)\\&=\lim_{\epsilon_1 \rightarrow 0} H_{\infty}^{\epsilon_1}(U^m|\bY^{\Delta},G_1(U^m),G_1,G_2(U^m),G_2) 
\end{align*}
To proceed, we lower bound $H_{\infty}^{\epsilon_1}(U^m|\bY^{\Delta},G_1(U^m),G_1,G_2(U^m),G_2)$ for a given $\epsilon_1>0$ (we specify $\epsilon_1$ later). Crucially, our lower bound will not depend on the quantization parameter $\Delta;$ this allows us to extend the same lower bound to the limiting quantity: $\lim_{\Delta\rightarrow 0} \lim_{\epsilon_1\rightarrow 0}H_{\infty}^{\epsilon_1}(U^m|\bY^{\Delta},G_1(U^m),G_1,G_2(U^m),G_2).$

We first recap (without proof) a few well known results.\\

\begin{claim}[Min-entropy~\cite{chain1}]\label{claim:min:entropy}
 For any $\mu,\mu',\mu_1,\mu_2 \in[0,1)$ and any set of jointly distributed discrete random variables $(X, Y, W)$, we have 
\begin{IEEEeqnarray}{rCl}
&&H_{\infty}^{\mu+\mu^{'}}(X,Y|W)-H_{\infty}^{\mu^{'}}(Y|W)\notag\\ 
&&\geq H_{\infty}^{\mu}(X|Y,W)\label{eq:min:1}\\ 
&&\geq H_{\infty}^{\mu_1}(X,Y|W)-H_0^{\mu_2}(Y|W)-\log\left[\frac{1}{\mu-\mu_1-\mu_2}\right]\label{eq:min:2}
\end{IEEEeqnarray}
\end{claim}
\begin{claim}[Max-entropy~\cite{chain1, chain2}]\label{claim:max:entropy}
For any $\mu,\mu',\mu_1,\mu_2 \in[0,1)$ and any set of jointly distributed random variables $(X, Y, W)$, we have 
\begin{IEEEeqnarray}{rCl}
&&H_{0}^{\mu+\mu^{'}}(X,Y|W)-H_{0}^{\mu^{'}}(Y|W) \notag\\
&&\leq H_{0}^{\mu}(X|Y,W)\label{eq:max:1}\\
&&\leq H_{0}^{\mu_1}(X,Y|W)-H_{\infty}^{\mu_2}(Y|W)+\log\left[\frac{1}{\mu-\mu_1-\mu_2}\right]\label{eq:max:2}
\end{IEEEeqnarray}
\end{claim}

We now state the following lemma:\\
\begin{lemma}\label{lem:smooth:min:entropy}
For any $\epsilon_1>0, \delta'>0$ and $n$ sufficiently large, 
\begin{align}
	H_{\infty}^{\epsilon_1}&(U^m|\bY^{\Delta},G_1(U^m),G_1,G_2(U^m),G_2)\notag
	\\&\stackrel{}{\geq} { n\left(\frac{1}{2}\log\left(\frac{P}{E}\right) -\frac{1}{2}\left( \log\left(1+\frac{P}{\gamma^2}  \right)\right)-\beta_1-\beta_2\right)}\notag \\& \hspace{35mm}-\log(\epsilon_1^{-1})-n\delta'\label{eq:h:inf:1}
    \end{align}
\end{lemma}
\begin{proof}
\begin{align}
	&H_{\infty}^{\epsilon_1}(U^m|\bY^{\Delta},G_1(U^m),G_1,G_2(U^m),G_2)\notag\\
	&\stackrel{(a)}{\geq} H_{\infty}(U^m,G_1(U^m),G_2(U^m)|\bY^{\Delta},G_1,G_2)\notag\\
	&\hspace{10mm}-H_{0}(G_1(U^m),G_2(U^m)|\bY^{\Delta},G_1,G_2)-\log(\epsilon_1^{-1})\notag\\
	&\stackrel{(b)}{=} H_{\infty}(U^m|\bY^{\Delta},G_1,G_2)\notag\\
	&\hspace{5mm}+H_{\infty}(G_1(U^m),G_2(U^m)|\bY^{\Delta},G_1,G_2,U^m)\notag\\
	&\hspace{10mm}-H_{0}(G_1(U^m),G_2(U^m)|\bY^{\Delta},G_1,G_2)-\log(\epsilon_1^{-1})\notag\\
	&\stackrel{(c)}{=} H_{\infty}(U^m|\bY^{\Delta},G_1,G_2)\notag\\
	&\hspace{10mm}-H_{0}(G_1(U^m),G_2(U^m)|\bY^{\Delta},G_1,G_2)- \log(\epsilon_1^{-1})\notag\\
	&\stackrel{(d)}{=} H_{\infty}(U^m|\bY^{\Delta},G_1,G_2)\notag\\
	&\hspace{10mm}-H_{0}(G_1(U^m),G_2(U^m)|\bY^{\Delta},G_1,G_2)- \log(\epsilon_1^{-1})\notag\\
	&\stackrel{(e)}{\geq}  H_{\infty}(U^m|\bY^{\Delta},G_1,G_2)\notag\\
	&\hspace{10mm}-H_{0}(G_1(U^m)|G_2(U^m),\bY^{\Delta},G_1,G_2)\notag\\
	&\hspace{20mm}-H_{0}(G_2(U^m)|\bY^{\Delta},G_1,G_2)-\log(\epsilon_1^{-1})\notag\\
	&\stackrel{(f)}{\geq} H_{\infty}(U^m|\bY^{\Delta},G_1,G_2)\notag\\
	&\hspace{4mm}-n\left(\bar{R}+\frac{1}{2}\log\left(\frac{E}{P}\right) + \beta_1\right)-n\beta_2-\log(\epsilon_1^{-1}) \label{eq:minE:1}
\end{align}
Here,

\begin{enumerate}[(a)]
\item from the chain rule for smooth min-entropy; see Claim~\ref{claim:min:entropy} and substitute $\mu=\epsilon_1$, $\mu_1=0$ and $\mu_2=0$ in~\eqref{eq:min:2}.
\item from the chain rule for min-entropy; see Claim~\ref{claim:min:entropy} and substitute $\mu=0$ and $\mu'=0$ in~\eqref{eq:min:1}.
\item from the fact that $G_1(U^m)$ and $G_2(U^m)$  are deterministic functions of $G_1$, $G_2$ and $U^m$. The quantity \\  $H_{\infty}(G_1(U^m),G_2(U^m)|\bY^{\Delta},G_1,G_2,U^m)=0$ irrespective of $\bY^{\Delta}.$
\item by the Markov chain  $\bX \leftrightarrow \bY \leftrightarrow (G_1,G_2)$.
\item from the chain rule for max-entropy; see Claim~\ref{claim:max:entropy} and  substitute $\mu=0$ and $\mu'=0$ in~\eqref{eq:max:1}. 
\item by noting that the  range of $G_1$ is $ \{0,1\}^{n(\bar{R}+\frac{1}{2}\log(\frac{E}{P})+ \beta_1)}$ and range of $G_2$ is $\{0,1\}^{n\beta_2}.$
\end{enumerate}

We now lower bound the first term in~\eqref{eq:minE:1}, i.e., $H_{\infty}(U^m|\bY^{\Delta},G_1,G_2).$ Here is the  lemma with the lower bound.
\begin{lemma}\label{lem:min-entropy:mutual}
For any $\delta'>0$ small enough and $n$ sufficiently large, we have
\begin{align}\label{eq:min-entropy:mutual}
H_{\infty}(U^m|\bY^{\Delta},G_1,G_2)\geq H(U^m)-I(U^m;\bY)-n\delta'.
\end{align}
\end{lemma}
\begin{IEEEproof}
To prove this result, we first recap the following known result which relates conditional smooth-min-entropy and conditional (Shannon) entropy. We use the specific version in~\cite{nascimento-barros-t-it2008} (cf.~\cite[Thm.~1]{nascimento-barros-t-it2008}).
\begin{theorem}[~\cite{nascimento-barros-t-it2008}]\label{thm:entropy:transfer}
    Let $P_{V^n,W^n}$ be a distribution over finite alphabets $\mathcal{V}^n\times\mathcal{W}^n.$ Then, for any constants $\delta',\epsilon'>0$ and $n$ sufficiently large, we have
    \begin{align}
        H_{\infty}^{\epsilon'}(U^n|V^n)\geq H(U^n|V^n)-n \delta'.
    \end{align}
\end{theorem}
\noindent We now simplify $H_{\infty}(U^m|\bY^{\Delta},G_1,G_2)$ as follows:
\begin{align}
H_{\infty}&(U^m|\bY^{\Delta},G_1,G_2)\\
&\stackrel{(a)}{=}\lim_{\epsilon'\rightarrow 0} H^{\epsilon'}_{\infty}(U^m|\bY^{\Delta},G_1,G_2)\notag\\
&\stackrel{(b)}{\geq } \lim_{\epsilon'\rightarrow 0} H(U^m|\bY^{\Delta},G_1,G_2)-n\delta'\notag\\
&\stackrel{}{=} H(U^m|\bY^{\Delta},G_1,G_2)-n\delta'\notag\\
&\stackrel{(c)}{=} H(U^m)-I(U^m;\bY^{\Delta},G_1,G_2)-n\delta' \label{eq:mutual:1}
\end{align}
where
\begin{enumerate}[(a)]
    \item follows from the definition of smooth-min-entropy.
    \item follows from Theorem~\ref{thm:entropy:transfer}.
    \item follows from chain rule of mutual information.
\end{enumerate}
Let us now simplify $I(U^m;\bY^{\Delta},G_1,G_2)$ in~\eqref{eq:mutual:1} as  $\Delta\rightarrow 0$.  Note that
\begin{align}
    \lim_{\Delta\rightarrow 0} I(U^m;\bY^{\Delta},G_1,G_2)&\stackrel{(a)}{=} I(U^m;\bY,G_1,G_2)\notag\\
    &\stackrel{(b)}{=}I(U^m;\bY)+I(U^m;G_1,G_2|\bY)\notag\\
    &\stackrel{(c)}{=}I(U^m;\bY). \label{eq:mutual:2}
\end{align}
where
\begin{enumerate}[(a)]
\item follows from definition of $\bY^{\Delta}$ and the mutual information $I(U^m;\bY^{\Delta},G_1,G_2)$ and their limiting values ($\Delta\rightarrow 0$).
\item follows from the chain rule of mutual information
\item follows from the Markov chain $U^m\leftrightarrow\bX\leftrightarrow \bY\leftrightarrow (G_1,G_2).$
\end{enumerate}
Putting together~\eqref{eq:mutual:1} and~\eqref{eq:mutual:2}, we have~\eqref{eq:min-entropy:mutual}.
This completes the proof of Lemma~\ref{lem:min-entropy:mutual}.
\end{IEEEproof}
Coming back to the main proof of Lemma~\ref{lem:smooth:min:entropy}, let us now simplify~\eqref{eq:minE:1} as follows:
\begin{align}
	&H_{\infty}^{\epsilon_1}(U^m|\bY^{\Delta},G_1(U^m),G_1,G_2(U^m),G_2)\notag\\
	&\stackrel{(a)}{\geq} \left(H(U^m)-I(U^m;\bY)-n\delta'\right)\notag\\
	&\hspace{20mm}-n\left(\bar{R}+\frac{1}{2}\log\left(\frac{E}{P}\right) + \beta_1\right)\notag\\
	&\hspace{50mm}-n\beta_2-\log(\epsilon_1^{-1}) \notag\\
	&\stackrel{(b)}{\geq} H(U^m)-I(\bX;\bY)-n\left(\bar{R}+\frac{1}{2}\log\left(\frac{E}{P}\right) + \beta_1\right)\notag\\
	&\hspace{50mm}-n\beta_2-\log(\epsilon_1^{-1})-n\delta' \notag\\
	&\stackrel{(c)}{\geq} H(U^m)-n\mathbb{C}_{AWGN}(\gamma^2)-n\left(\bar{R}+\frac{1}{2}\log\left(\frac{E}{P}\right) + \beta_1\right)\notag\\
	&\hspace{40mm}-n\beta_2-\log(\epsilon_1^{-1})-n\delta' \notag\\
	&\stackrel{(d)}{=} n\bar{R}-n\left(\frac{1}{2}\log\left(1+\frac{P}{\gamma^2}\right)\right)\notag\\
	&\hspace{20mm}-n\left(\bar{R}+\frac{1}{2}\log\left(\frac{E}{P}\right) + \beta_1\right)\notag\\
	&\hspace{40mm}-n\beta_2-\log(\epsilon_1^{-1})-n\delta' \notag\\
	&\stackrel{}{=} n\left(\bar{R}-\frac{1}{2}\log\left(1+\frac{P}{\gamma^2}\right)\right)-n\left(\bar{R}+\frac{1}{2}\log\left(\frac{E}{P}\right) + \beta_1\right)\notag\\
	&\hspace{40mm}-n\beta_2-\log(\epsilon_1^{-1})-n\delta' \notag\\
	&\stackrel{(e)}{=} n\left(\frac{1}{2}\log\left(\frac{P}{E}\right)-\frac{1}{2}\log\left(1+\frac{P}{\gamma^2}\right)\right)-n\left(\beta_1+\beta_2\right)\notag\\
	&\hspace{60mm}-\log(\epsilon_1^{-1})-n\delta' \notag\\
\end{align}

\begin{enumerate}[(a)]
    \item follows from Lemma~\ref{lem:min-entropy:mutual}.
    \item follows from the Markov chain $U^m\leftrightarrow \bX\leftrightarrow \bY$ and the data processing inequality.
    \item follows from noting that $I(\bX;\bY)\leq n \mathbb{C}_{AWGN}(\gamma^2)$ where $\mathbb{C}_{AWGN}(\gamma^2):=\frac{1}{2}\log\left(1+\frac{P}{\gamma^2}\right)$ is the communication capacity of an AWGN  channel with noise variance $\gamma^2$ under  input power constraint $P$. Note that we need to allow the possibility that a cheating Bob may privately  fix an AWGN channel where the  variance may take any value in the range $[\gamma^2,\delta^2].$
    \item follows from noting that $H(U^m)=n\bar{R}$ and substituting for $\mathbb{C}_{AWGN}(\gamma^2).$
    \item follows from cancelling the term $n\bar{R}$ and rearranging the terms.
\end{enumerate} 
\end{proof}

Since the lower bound does not depend on $\Delta>0,$ the following lemma is straight forward. Note the change to the continuous random vector $\bY$ (instead of $\bY^{\Delta}$ as in the previous lemma) as part of Bob's view.
\begin{lemma}\label{lem:min:entropy:lb}
For any $\epsilon_1>0, \delta'>0$ and $n$ sufficiently large, 
\begin{align}
	H_{\infty}^{\epsilon_1}&(U^m|\bY,G_1(U^m),G_1,G_2(U^m),G_2)
	\notag\\
	&\hspace{5mm}\stackrel{}{\geq} { n\left(\frac{1}{2}\log\left(\frac{P}{E}\right) -\frac{1}{2}\log\left(1+\frac{P}{\gamma^2}\right)-\beta_1-\beta_2\right)}\notag\\
	&\hspace{50mm}-\log(\epsilon_1^{-1})-n\delta'\label{eq:h:inf:1}
    \end{align}
\end{lemma}

Next, we use Lemma~\ref{lem:glh} to show that the distribution of the secret key $\text{Ext}(\bX)$ is statistically close to  a uniform distribution thereby achieving bias-based secrecy. Let us fix $\epsilon_1:=2^{-n\alpha_2}$, where $\alpha_2>0$ is an arbitrary small constant.
We make the following correspondence in  Lemma~\ref{lem:glh}: $G\leftrightarrow \text{Ext}$, $W\leftrightarrow U^m$ and  $l\leftrightarrow nR$
%$I\leftrightarrow (\bY,G_1(\bX),G_1,G_2(\bX),G_2)$ 
to get the following:
\begin{align}
    \text{SD}(&P_{\text{Ext}(U^m),\text{Ext}},P_{U_l,\text{Ext}}) \notag \\
    &\stackrel{(a)}{\leq} \frac{1}{2}\sqrt{2^{ -H_{\infty}(U^m)} 2{^{nR}}}\notag\\
		&\stackrel{(b)}{\leq} \frac{1}{2}\sqrt{2^{ -H_{\infty}(U^m|\bY^{\Delta},G_1(U^m),G_1,G_2(U^m),G_2)} 2^{nR}}\notag\\
    &\stackrel{(c)}{\leq}\frac{1}{2}
	\sqrt{2^{-n\left(\frac{1}{2}\log\left(\frac{P}{E}\right) -\frac{1}{2}\left( \log\left(1+\frac{P}{\gamma^2}  \right)\right)-\beta_1-\beta_2-\alpha_2-\delta'\right)}} \notag\\
	&\hspace{20mm} \cdot \sqrt{2^{n\left(\frac{1}{2}\log\left(\frac{P}{E}\right) -\frac{1}{2}\left( \log\left(1+\frac{P}{\gamma^2}  \right)\right)   -\beta_3\right)}}\notag\\
    &= \frac{1}{2}\sqrt{2^{n(\beta_1+\beta_2+\alpha_2+\delta'-\beta_3)}}\notag\\
    %&\leq \frac{1}{2}\sqrt{2^{n(\gamma_1+\gamma_2-\gamma_3)}}\\
    &\stackrel{(d)}{\leq} 2^{-n\alpha_3}  \label{eq:sd:1}
\end{align}
where, $n$ is sufficiently large so that $\delta'>0$ is negligibly small such that $\alpha_3>0.$
Here,\begin{enumerate}[(a)]
\item follows directly from Lemma~\ref{lem:glh}.
\item follows as conditional min-entropy (under any $\Delta>0$ sufficiently small) lower bounds min-entropy. This also holds under the limit $\Delta\rightarrow 0.$
\item follows from the definition of $R$ (cf.~\eqref{eq:commit:rate}) and Lemma~\ref{lem:min:entropy:lb} 
\item follows from noting that $\beta_3$ is chosen such that $\delta'+\beta_1+\beta_2+\alpha_2-\beta_3<0$; here, we note that $\alpha_2$ is an arbitrarily chosen (small enough) constant, and $\delta'>0$  can be made arbitrarily small for $n$ sufficiently large. As such, a choice of $\beta_3>\beta_1 +\beta_2$ is sufficient.
\end{enumerate}
From~\eqref{eq:sd:1} and Lemma~\ref{lem:glh}, it follows that we can extract $n\left(\frac{1}{2}\log\left(\frac{P}{E}\right) -\frac{1}{2}\left( \log\left(1+\frac{P}{\gamma^2}  \right)\right)-\beta_3\right)$ almost uniformly random bits which proves the security of the secret key; this guarantees that our commitment protocol satisfies bias-based secrecy (cf.~\cite[Def.~3.1]{damgard1998statistical}).  

To conclude the concealment analysis, recall from our discussion earlier (see also~\cite[Th.~4.1]{damgard1998statistical}) that bias-based secrecy under \emph{exponentially decaying} statistical distance, as in~\eqref{eq:sd:1}, implies capacity-based secrecy. Since we have already shown that the protocol satisfies bias-based secrecy with exponentially decaying security parameter, hence, the protocol satisfies capacity-based secrecy. In particular, for $n$ sufficiently large,  $I(C;\view_B)\leq \epsilon$ and our protocol is $\epsilon$-concealing. \\

\noindent \emph{\textbf{\underline{(iii). $\epsilon-$ binding:}}}\\

%An $\epsilon-$binding commitment protocol allows Bob to catch any dishonest behaviour of Alice by verifying that the revealed bit string $\td{\bc}$ is same as the committed bit string $\bc$. 
%
%Here, we analyse the `most favourable' (i.e., worst-case w.r.t. the protocol) cheating strategy of Alice and show that our protocol prevents Alice from cheating successfully.
%
\removed{%Alice figure
\begin{figure}[!ht]
        \centering
        \includegraphics[scale=0.5]{bind2.jpg}
        \caption{Alice's cheating strategy:  set REC$[\gamma,\delta]$ to BSC($s$), $s\in[\gamma,\delta]$ and send $\bx$ in the commit phase, and then reveal  $\bx'\neq \bx$ in the reveal phase. It can be shown that the `best' such choice for Alice is to fix $s=\gamma$.}
        \label{fig:my_label}
\end{figure}
}%removed
%\noindent\textbf{Alice's  cheating strategy:} 
To analyse binding, we analyse the scenario where a potentially dishonest Alice seeks to confuse Bob between two (or more) different commit bit strings in $\{0,1\}^m,$  say $\bar{u}^m$ and $\td{u}^m$ (i.e., Bob's test accepts two different commit strings). We seek to show that w.h.p our commitment protocol precludes any such possibility.

To begin, a cheating Alice seeks to maximize the set of potential bit strings in $\{0,1\}^m$ that would appear potential candidates in the list $\cL(\by)$ generated by Bob. Toward the same, a cheating Alice employs the following strategy: she first picks up a vector $\bx\in \cS(0,\sqrt{n(P-\gamma^2)})$ in the commit phase. Next, if actively dishonest she may privately fix the variance of the  \gunc \ to any value $s^2\in[\gamma^2,\delta^2].$ It will be apparent later that (cf. Claim~\ref{claim:A:set}) that the `worst' such choice would be the lowest value possible, i.e., $s^2=\gamma^2.$  Let us define $E_s:=\delta^2-s^2$. Note that $E=E_{\gamma}=\delta^2-\gamma^2$.
%\aky{since, unlike UNC a REC doesn't comes with any $\gamma$, it only comes into picture, when Alice starts behaving dishonestly and sets the $BSC(\delta)$ to any arbitrary $\gamma$ so, should we avoid this idea of Alice's behavioir as choosing some $s$ and then optimizing to $\gamma$ ? }
Let $\bX=\bx$ be the transmitted vector and $\bY=\by$ be the  bit string received by Bob's over the AWGN($s^2$). Note that a cheating Alice need not transmit a codeword, however $\bx\in\cS(P),$ i.e., the transmitted vector needs to satisfy the transmit power constraint $P.$ Alice can cheat successfully by confusing Bob in the reveal phase only if she can find two \emph{distinct} length-$m$ binary strings, say $\bar{u}^m$ and $\td{u}^m$ such that (i) if $\psi(\bar{u}^m)=\bar{\bx}$ and $\psi(\td{u}^m)=\td{\bx}$ then  $\bx',\td{\bx} \in \cL(\by)$, and (ii) $\bar{u}^m$ and $\td{u}^m$ pass the two rounds of sequential random hash exchange challenge (w.r.t hash functions $G_1(\cdot)$ and $G_2(\cdot)$).
Let $\cA$ denote all codewords in $\cC$ corresponding to such  length-$m$  bit strings. Then, the following claim shows that $\cA$ can be exponentially large.\\

\begin{claim}\label{claim:A:set}
Given any $\eta>0$, for $n$ sufficiently large,
\begin{equation}
|\cA|\leq 2^{n(\bar{R}+\frac{1}{2}\log(\frac{E}{P})+\eta)}.
\end{equation}
\end{claim}
The proof appears in Appendix~\ref{app:A:set}. Note that, essentially, from  the above claim, one can conclude that the choice of $s=\gamma^2$ is the `best' choice for a cheating Alice (such a choice maximizes $|\cA|$), i.e., Alice can be no worse than when it privately fixes the \guncw\gsqdsq\ to an AWGN channel with variance $\gamma^2$. We will choose  $0<\eta<\beta_1$ later (cf. Claim~\ref{claim:step:1}). 

We now show that our choice of hash functions $G_1(\cdot)$ and $G_2(\cdot)$ allows us to essentially `trim' down this set $\cA$ of `confusable' vectors all the way down to none (this will prevent a cheating Alice from confusing Bob with more than 1 commit strings).  Recall that Alice's choice in the commit phase is $\bx$. For a given hash value $h_1\in\{0,1\}^{n(\bar{R}+\frac{1}{2}\log(\frac{E}{P}) + \beta_1)}$ sent by Alice, let 
\begin{equation}
I_i(h_1):=\begin{cases}
1& \mbox{ if } G_1(u^m_i)=h_1 \\
0& \mbox{ otherwise.}
\end{cases}
\end{equation}
$I_i(h_1)$ is an indicator random variable which identifies if $u^m_i$ has a \emph{hash-collision} under $G_1$ with the hash value $h_1.$
Also, let 
\begin{align}
I(h_1):=\sum_{i=1}^{|\cA|} I_i(h_1)
\end{align}
 denotes the total number of hash collisions with hash value $h_1$.
%\aky{Everywhere below shouldn't we use $g_1(.)$ instead of $G_1(.)$ as choice of $G$, is already fixed by Bob in the commit phase, and shouldn't the expectations below be over $x$'s within the set $\cA$, instead of $G$?}
Then, the following holds when $0<\eta<\beta_1$ (see proof in Appendix \ref{app:step:1}):\\

\begin{claim}\label{claim:step:1}
$$\mathbb{P}\left(\exists\ h_1\in \{0,1\}^{n(\bar{R}+\frac{1}{2}\log(\frac{E}{P})+ \beta_1)}:  I(h_1)>6n\bar{R}+1\right)$$ vanishes exponentially in $n$ as $n\to\infty$.
\end{claim}
This implies that the size of the `confusable' set \emph{after} the first hash challenge via $G_1$ for any $h_1$ is larger that $n\bar{R}$ (i.e., linear in blocklength $n$) with only exponentially small probability (in block length $n$). 

Conditioned on the event $I(h_1)\leq 6 n\bar{R}+1$, $\forall\ h_1$, which occurs with high probability (w.h.p.), we now analyse the size of the `confusable' set \emph{after} the second hash challenge via $G_2$; let $\mathcal{F}_{h_1}$ denote this set of `confusable' vectors after the second hash challenge for a given $h_1$. We prove the following claim (proof in Appendix~\ref{app:step:2}):\\
\begin{claim}\label{claim:step:2}
For every $h_1\in\{0,1\}^{n(\bar{R}+\frac{1}{2}\log(\frac{E}{P}) + \beta_1)}$, we have for $n$ sufficiently large 
\begin{IEEEeqnarray*}{rCl}
\mathbb{P}\left(\exists\ \bx\neq \bx'\in\mathcal{F}_{h_1} :G_2(u^m)=G_2(u'^m)\big| I(h_1)\leq  6n\bar{R}+1\right)\leq 2^{-n\frac{\beta_2}{2}} 
\end{IEEEeqnarray*}
\end{claim}
As the above claim holds for every $h_1$, and noting that $\beta_2>0$, we now choose $n$ large enough to conclude that our commitment protocol is $\epsilon-$binding.

\subsection{Proof of Claim~\ref{claim:A:set}}\label{app:A:set}
From the definition of $\cA$, we have
\begin{align}
|\cA|&\stackrel{(a)}{\leq} 2^{n\left(\bar{R}+\frac{1}{2}\log\left(\frac{E_s}{P}\right)+\eta\right)}\notag\\
 &\stackrel{(b)}{\leq} 2^{n\left(\bar{R}+\frac{1}{2}\log\left(\frac{E}{P}\right)+\eta\right)}
\end{align}
where 
\begin{enumerate}[(a)]
\item follows from noting that an honest Bob will accept a vector $\bx'$ if $\bx'\in\cL(\bY)$; since a cheating may privately fix the variance of  Gaussian UNC$[\gamma,\delta]$ to some $s^2\in[\gamma^2,\delta^2]$ resulting in elasticity $E_s$, the total number of such confusable codebook vectors are at most $2^{n\left(\bar{R}+\frac{1}{2}\log\left(\frac{E_s}{P}\right)+\eta\right)}$, where $\eta>0$ choice can be arbitrary, when blocklength $n$  is sufficiently large. 
\item follows from noting that $E_s\leq E=\delta^2-\gamma^2<P$ which results in a potentially larger (exponential size) set $\cA.$  
\end{enumerate}
%This concludes the proof of the claim.
%
\subsection{Proof of Claim~\ref{claim:step:1}}\label{app:step:1}
The proof of this claim follows by standard concentration techniques. We first bound the expected number of hash-collisions $\mathbf{E}_{G_1}[I(h_1)]$ for a given hash value $h_1.$ In particular, we show that for $n$ large enough, the expected number of such collisions $\mathbf{E}_{G_1}[I(h_1)]<1.$ We now concentrate using this expected value and identify the `bad' hash values, say $h'$, where the expected number of hash collisions $\mathbf{E}_{G_1}[I(h')]$ exceeds the average value by  a `non-trivial' amount.
\balance
As $G_1\sim \text{Unif}\left(\mathcal{G}_1\right)$, we have
%\begin{align}
$\mathbf{E}_{G_1}[I(h_1)]\leq \sum_{i=1}^{|\cA|}2^{-(n(\bar{R}+\frac{1}{2}\log(\frac{E}{P})+\beta_1))}\leq 2^{n(\eta-\beta_1)},$ where the final inequality follows  from~Claim~\ref{claim:A:set} and noting that $\beta_1> \eta.$ We set $\td{\beta}_1:=\beta_1-\eta>0$  to get $\mathbf{E}_{G_1}[I(h_1)]{\leq}2^{-n\td{\beta}_1}.$ Hence, for $n$ sufficiently large, we have $\mathbb{E}[I(h_1)]\leq 1$, $\forall h_1$.
%\end{align} 
%which is independent of $h_1$. Here $(a)$ follows from the definition of $\mathcal{G}_1$, $(b)$ follows from~Claim~\ref{claim:A:set} and noting that $\beta_1> \eta$; letting $\td{\beta}_1:=\beta_1-\eta>0$  gives us $(c)$. Note that for $n$ sufficiently large, we have $\mathbb{E}[I(h_1)]\leq 1$, $\forall h_1$. 
%
\removed{
Recall that $G_1\sim \text{Unif}\left(\mathcal{G}_1\right)$. Then,
\begin{align}
\mathbf{E}_{G_1}[I(h_1)]
%&\stackrel{}{=} \mathbf{E}_{G_1}\left[\sum_{i=1}^{|\cA|} I_i(h_1)\right] \notag \\
%&\stackrel{}{=}\sum_{i=1}^{|\cA|}\mathbf{E}_{G_1}[I_i]\notag \\
%&\stackrel{}{=}\sum_{i=1}^{|\cA|}\mathbb{P}_{G_1}\left(G_1(\bx_i)=G_1(\bx)=h_1\right) \notag \\
&\stackrel{(a)}{\leq}\sum_{i=1}^{|\cA|}2^{-(n(H(E)+\beta_1))} \notag \\
%&\stackrel{(b)}{\leq} 2^{n(H(E)+\eta)-(n(H(E)+\beta_1))} \notag \\
&\stackrel{(b)}{\leq}2^{n(\eta-\beta_1)}\notag\\
&\stackrel{(c)}{\leq}2^{-n\td{\beta}_1}\label{eq:I}
\end{align} 
which is independent of $h_1$. Here $(a)$ follows from the definition of $\mathcal{G}_1$, $(b)$ follows from~Claim~\ref{claim:A:set} and noting that $\beta_1> \eta$; letting $\td{\beta}_1:=\beta_1-\eta>0$  gives us $(c)$. Note that for $n$ sufficiently large, we have $\mathbb{E}[I(h_1)]\leq 1$, $\forall h_1$. 
}%removed
We  need the following result to proceed:\\

\begin{lemma}[~\cite{rompel}]\label{lem:rompel}
	Let $X_1,X_2,X_3....X_m\in [0,1]$ be $l$-wise independent random variables, where $l$ is an even and positive integer. Let  $X:=\sum_{i=1}^{m} X_i$, $\mu:=\mathbf{E}[X]$, and $\kappa>0$ be a constant. Then,
	\begin{align}
	\mathbb{P}\left(|X-\mu|>\Delta\right)<O\left(\left(\frac{l\mu+l^2}{\kappa^2}\right)^{l/2}\right)    
	\end{align}
\end{lemma}
We make the following correspondence: $l\leftrightarrow 3n\bar{R}$, $\kappa\leftrightarrow 2l = 6n\bar{R}$. Then, using the union bound, we get:
\begin{align}
\mathbb{P}&\left(\exists h_1\in \{0,1\}^{n(\bar{R}+\frac{1}{2}\log(\frac{E}{P})+ \beta_1)}: I(h_1)>6n\bar{R}+1\right) \\
&\leq  \sum_{h_1\in \{0,1\}^{n(\bar{R}+\frac{1}{2}\log(\frac{E}{P})+ \beta_1)} } \mathbb{P}\left(I(h_1)>6n\bar{R}+1\right)\\
&\stackrel{(a)}{\leq}  2^{n(\bar{R}+\frac{1}{2}\log(\frac{E}{P})+\beta_1)}O\Big(\Big(\frac{l\mu+l^2}{\kappa^2}\Big)^{l/2}\Big) \notag\\
&\stackrel{(b)}{\leq} 2^{n(\bar{R}+\frac{1}{2}\log(\frac{E}{P})+\beta_1)} O\Big(\Big(\frac{1+l}{4l}\Big)^{l/2}\Big) \notag\\
&<2^{n(\bar{R}+\frac{1}{2}\log(\frac{E}{P})+\beta_1)} O(2^{-l/2}) \notag\\
&=2^{n(\bar{R}+\frac{1}{2}\log(\frac{E}{P})+\beta_1)} O(2^{-\frac{3}{2}n\bar{R}}) \label{eq:aa} 
\end{align}
where we have
\begin{enumerate}[(a)]
\item from Lemma~\ref{lem:rompel}
\item by noting that for $n$ sufficiently large,  $\mu=\mathbb{E}[I(h_1)]\leq 1$, $\forall h_1$, and making the correspondence $\kappa\leftrightarrow 2k$.
\end{enumerate}
Now note that~\eqref{eq:aa} tends to zero exponentially fast as we have $(\bar{R}+\frac{1}{2}\log(\frac{E}{P})+\beta_1)< \frac{3}{2}\bar{R}$. 
% This completes the proof of claim.
%
\subsection{Proof of Claim~\ref{claim:step:2}}\label{app:step:2}
Recall the definition of $\mathcal{F}_{h_1}$, and let $\mathcal{F}:=\max_{h_1} \mathcal{F}_{h_1}$. From Claim~\ref{claim:step:1}, $|\mathcal{F}|\leq 6n\bar{R}+1$ with exponentially vanishing probability of error. Noting that $G_2\sim\text{Unif}\left(\mathcal{G}_2\right)$, where $\mathcal{G}_2=\{g_2:\{0,1\}^n \rightarrow \{0,1\}^{n\beta_2}\}$, we have for every $h_1\in\{0,1\}^{n(\bar{R}+\frac{1}{2}\log(\frac{E}{P})+\beta_1)}$,
%\aky{should $g_2$ be used instead of $G_2$?}
%In the second round of hash exchange where a 2-universal hash function is used, the probability of collision between any $2$ of the $8n+1$ elements, say $g_2(x)$ and $g_2(x')$ is upper bounded by
\begin{align}
\mathbb{P}&\left(\exists \bx\neq \bx'\in\mathcal{F}_{h_1}:G_2(\bx)=G_2(\bx')\big| I(h_1)\leq 6n\bar{R}+1\right)\notag\\
&\stackrel{(a)}{\leq} {\mathcal{F}\choose {2}} \mathbb{P}\left(G_2(\bx)=G_2(\bx')\right) \notag\\
&\stackrel{(b)}{\leq} {{6n\bar{R}+1}\choose {2}}2^{-n\beta_2} \notag\\
%&< (8n+1)(8n)2^{-n\beta_2} \notag\\
&\leq 2^{-n\frac{\beta_2}{2}} \hspace{5mm} \text{ for $n$ large enough},
\end{align}
where $(a)$ follows from the definition of $\mathcal{F}$, and using the union bound (on distinct pairs of vectors in $\mathcal{F}$); we get $(b)$ from the definition of $\mathcal{G}_2$. \\
This completes the proof of the claim.

\end{document}